\documentclass[12pt]{article}

\usepackage{amsmath}
\usepackage{amsfonts}
\usepackage{latexsym}
\usepackage{graphicx}
\usepackage{amssymb}
\usepackage{amsthm}
\usepackage[margin=1.8cm]{geometry}
\usepackage{url}
\usepackage{color}
\usepackage{enumerate}
\usepackage[small]{caption}
\usepackage[shortlabels]{enumitem} 
\usepackage{tikz}
\usetikzlibrary{matrix,arrows}

\usepackage[T1]{fontenc}     
\usepackage{lmodern}         

\usepackage[utf8]{inputenc}  

\newtheorem{definition}{Definition}
\newtheorem{lemma}[definition]{Lemma}
\newtheorem{proposition}[definition]{Proposition}
\newtheorem{example}[definition]{Example}

\newtheorem{theorem}[definition]{Theorem}
\newtheorem{remark}[definition]{Remark}

\newcommand{\NN}{\mathbb{N}}
\newcommand{\ZZ}{\mathbb{Z}}
\newcommand{\QQ}{\mathbb{Q}}
\newcommand{\RR}{\mathbb{R}}

\newcommand{\Card}{\mathrm{Card}}
\newcommand{\history}{\mathrm{history}}
\renewcommand{\S}{\mathcal{S}}
\renewcommand{\L}{\mathcal{L}}
\newcommand{\A}{\mathcal{A}}
\newcommand{\G}{\mathcal{G}}
\newcommand{\tilP}{\widetilde{\mathcal{P}}}
\newcommand{\be}{\ensuremath{\mathbf e}}
\newcommand{\bu}{\ensuremath{\mathbf u}}

\newcommand{\bx}{\ensuremath{\mathbf x}}

\newcommand{\vect}[1]{\overrightarrow{#1}}
\newcommand{\emptyword}{\varepsilon}

\newcommand{\age}{\mathrm{age}}
\newcommand{\yes}{\mathrm{yes}}
\newcommand{\no}{\mathrm{no}}

\begin{document}

\title{Factor Complexity of $S$-adic sequences generated by the 
Arnoux-Rauzy-Poincaré Algorithm}
\author{{\sc V. Berth\'e and S. Labb\'e}\\  \\
\small LIAFA, Université Paris Diderot - Paris 7,\\ [-0.6ex]
\small Case 7014, 75205 Paris Cedex 13, France\\ [-0.6ex]
\small \tt berthe@liafa.univ-paris-diderot.fr\\ [-0.6ex]
\small \tt labbe@liafa.univ-paris-diderot.fr}
\date{\small Mathematics Subject Classifications: 68R15, 37B10.}

\maketitle

\begin{abstract}

The Arnoux-Rauzy-Poincar\'e    multidimensional  continued  fraction algorithm  is    
 obtained   by combining the 
 Arnoux-Rauzy and Poincar\'e  algorithms. It is a generalized Euclidean  algorithm.
 Its three-dimensional   linear version  consists in
 subtracting the sum  of the  two smallest entries to the largest if possible (Arnoux-Rauzy step),
and otherwise,  in subtracting  the smallest entry to the median and the median to the
largest (the Poincar\'e step), 	and  by performing when possible Arnoux-Rauzy steps in priority. After  renormalization it provides  a  piecewise fractional map of the  standard  $2$-simplex.
 We  study   here the factor complexity of its  associated
 symbolic dynamical system, defined as  an $S$-adic system.   It is   made of infinite words generated by
the composition  of   sequences of  finitely many substitutions, together with some restrictions
concerning  the allowed  sequences  of substitutions  expressed in terms of a
regular language. Here, the   substitutions   are provided  by the matrices of the linear  version of the  algorithm.  We give an 
 upper bound for  the linear growth   of the  factor complexity.  We then deduce  the 
convergence of the associated algorithm by unique ergodicity.

%
%

\end{abstract}

\section{Introduction}

Multidimensional continued fraction algorithms aim at  providing  good rational approximations
of a  given vector. 
There exist  many different  types of
continued fraction algorithms. Among them, piecewise fractional ones in the 
sense of  \cite{BRENTJES,SCH} have been widely studied   whereas for their    arithmetic or for their ergodic properties.
The viewpoint we take here on these algorithms is issued from word combinatorics and symbolic dynamics.
It is indeed possible to  generate with such algorithms  infinite words  with prescribed letter frequencies:   the letter  frequency vector 
is indeed  the vector on which the algorithm is applied.
We recall that   a substitution is a morphism of the  free monoid that replaces letters by finite words.
A piecewise fractional   continued fraction algorithm  produces  (unimodular) matrices with non-negative entries that we consider  as incidence matrices of   substitutions.
We then iterate these substitutions in an  $S$-adic way, that is,  as the  (inverse)  limit of an infinite product of substitutions (see e.g.   \cite{BDRIMS,MR2759107,durand_do_2013,Leroy12}). We  thus obtain an infinite word
$\bu$ of the form 
$$\bu= \lim_{n\to\infty}\sigma_0\circ\sigma_1\circ\cdots\circ\sigma_n(a^{\infty}).$$
As an  illustration consider the generation  of  Sturmian words    with the  classical continued fraction algorithm (see  \cite{lcfal,fogg_substitutions_2002} for more details). 

The  Diophantine approximation   properties  of the underlying continued fraction algorithm are  reflected in the generic behaviour of the   balance function of 
  the  generated word  $\bu$,  where the  balance   function  counts, for each given letter,
the difference  between the numbers of occurrences of this letter in any  two  words of the same length that occur in  $\bu$. It is also  closely related to the notion
of symbolic discrepancy such as  considered in \cite{Adam03}. We also    would like  the combinatorics of      the  generated infinite word
$\bu$  to be ``simple'' in the sense that  the factor complexity of  $\bu$ 
is expected to  be of
 linear   growth, were the factor complexity counts the number of  factors of  a given length.

Observe that there exist several  methods  for  producing  infinite   words with    prescribed  letter  frequencies  having  a linear factor complexity $p(n)$  and/or a bounded balance. 
The Sturmian words form a well-known family of infinite balanced words over a two-letter alphabet having
a linear factor complexity ($p(n)=n+1$ for all $n$).  Nevertheless the situation is more contrasted
for words defined on  alphabets having at least three letters concerning the possibility of having simultaneously   prescribed  letter  frequencies, a linear factor complexity and  a bounded balance.
Typical  generalizations of Sturmian words  are natural  codings of  interval exchanges  and the billiard words  in the $d$-dimensional cube.  However, billiard words
have quadratic factor complexity \cite{MR1372799,MR1963971} and 
codings of interval exchanges  are not balanced \cite{MR1488330}. 
Other approaches were considered in digital geometry where
arithmetic definitions of $3$D discrete lines were proposed.  The standard
model of  \cite{andres} is one of them and can be encoded as a word on
a three-letter alphabet. It turns out that this model  corresponds to the one of 
billiard words \cite{labbe_structure_2012}, thus also having a quadratic factor
complexity in general.  

The experimentations described in   \cite{EPTCS63.8,labbe_structure_2012} indicate  that 
some  multidimensional continued fraction algorithms  generate  $S$-adic words
 having  a linear factor complexity and  a bounded balance for almost every    letter frequencies vector.  In particular,  Brun multidimensional continued
fraction algorithm as well as the  Arnoux-Rauzy-Poincar\'e algorithm  seem to be the two best choices  in terms of balance properties. In
this article, we focus on the Arnoux-Rauzy-Poincar\'e algorithm   which  performs  experimentally a bit better
than does  Brun  algorithm.  
This algorithm (under its linear form) consists in subtracting  the sum  of the  two smallest entries to the largest if possible
and otherwise, in  subtracting  the smallest entry to the median and the median to the
largest.
In order to generate infinite words, 
we  introduce an  $\S$-adic system associated with  the nine possible   matrices   of the algorithm  that  thus provide  
 a set $\S$ of nine substitutions. Three of them are substitutions known under the name of
Arnoux-Rauzy substitutions \cite{arnoux_representation_1991}, and the other
six are named  Poincar\'e substitutions after Poincar\'e  algorithm
\cite{nogueira95}.  Moreover, the execution of the Arnoux-Rauzy-Poincar\'e algorithm
yields   restrictions to the allowed   infinite sequences of substitutions,
expressed in terms of a  regular language.  
We then have a bijection
(up to a set of zero measure)  between  the infinite words in the  corresponding
$\S$-adic system and the  standard $2$-simplex  $\Delta=\{(x_1,x_2,x_3)\in\RR^3_+ \mid x_1+x_2+x_3=1\}$
(the vectors of  letter frequencies).
The main result of the present  paper is  that  these words   have a  linear
factor complexity $p(n)$.

\begin{theorem}[\bf Factor Complexity]\label{thm:leq3n}
Let $\bu$ be an $\S$-adic word generated by the Arnoux-Rauzy-Poincar\'e
algorithm applied to  a totally irrational vector $\bx \in \Delta$. Then the factor complexity of $\bu$ is such that 
$p(n+1)-p(n)\in\{2,3\}$ and
$2n+1\leq p(n)\leq \frac{5}{2}n+1$
for all $n\geq 0$.
\end{theorem}


The proof relies on a careful study  of  bispecial factors of $\bu$, that is, of factors having several  left and right extensions in $\bu$.   We prove that weak and strong bispecial factors are
alternating in the sequence (ordered by increasing length) of non-neutral
bispecial factors.  
The restriction  for the  directive sequences of  the $\S$-adic words to the regular language  provided by the
Arnoux-Rauzy-Poincar\'e algorithm is clearly important;  indeed  quadratic  factor complexity   can be reached otherwise (see Section \ref{sec:quadratic}).

Then,  by using a result of
Boshernitzan \cite{Boshernitzan1984}, we deduce unique ergodicity and thus,   the existence of  (uniform)
frequency of any factor, and in particular of the letters.   This  also  provides
a combinatorial proof of convergence  for this multidimensional continued
fraction algorithm.

\begin{theorem}[\bf Frequencies and Convergence]\label{thm:conv}
Let $\bu$ be an $\S$-adic word generated by the Arnoux-Rauzy-Poincar\'e
algorithm applied to a totally irrational vector $\bx \in \Delta$.
Then   the symbolic  dynamical system  generated by $\bu$ is uniquely ergodic.
As a consequence, the frequencies of factors and letters  exist in $\bu$, the
latter being equal to the coordinates of  $\bx$.
 
Furthermore, the  Arnoux-Rauzy-Poincar\'e algorithm  is a weakly  convergent
algorithm, that is,  for   Lebesgue almost   every $\bx \in \Delta $,  if
$(M_n)_n$ stands for the  sequence of matrices  produced by the
Arnoux-Rauzy-Poincar\'e algorithm, then one has 
$\cap_n M_0 \cdots M_n ({\mathbb R}_+^3)={ \mathbb R}_ + \bx.$
\end{theorem}

Let us sketch the content of the present paper.  
The Arnoux-Rauzy-Poincar\'e multidimensional continued fraction algorithm  is introduced in 
Section~\ref{sec:mcf}. We   also define the  associated  Arnoux-Rauzy-Poincar\'e $\S$-adic system based  on its
nine  substitutions (provided by its linear version)  together with a  rational restriction on the directive   sequences of substitutions that are iterated.
In Section~\ref{sec:prelim}, we introduce the basic notions used to compute the
factor complexity, namely languages, bispecial factors and extension types.
In Section~\ref{sec:arpbispecials},
we study the life of bispecial factors under Arnoux-Rauzy and Poincar\'e substitutions
(with no restriction on the order of  application of substitutions).  
In Section~\ref{sec:proofthm}, we prove the upper  bound on the factor complexity    stated in Theorem~\ref{thm:leq3n}.
The convergence of the algorithm  together with  unique ergodicity  is   lastly considered in
 Section~\ref{sec:convergence}.
 \medskip

This article is an extended version of \cite{DBLP:conf/cwords/BertheL13}.  The present paper provides    the upper bound
$p(n) \leq \frac{5}{2} n+1$, whereas  the upper bound in  \cite{DBLP:conf/cwords/BertheL13} was $p(n) \leq 3n+1$.

\bigskip

{\bf Acknowledgements} We are thankful to   Pierre Arnoux, Srecko Brlek, Julien Cassaigne,  Julien Leroy and Thierry Monteil   for
many fruitful discussions on the subject. This work was supported by Agence Nationale de la Recherche and the Austrian Science Fund
through project Fractals and Numeration ANR-12-IS01-0002 and project Dyna3S ANR-13-BS02-0003.
The second author is supported by NSERC (Canada).

\section{The Arnoux-Rauzy-Poincar\'e Algorithm}\label{sec:mcf}

\subsection{The algorithm}\label{subsec:algo}

The Arnoux-Rauzy-Poincar\'e (\textbf{ARP}) is  a  multidimensional continued
fraction algorithm in the sense of~\cite{BRENTJES,SCH}, defined by piecewise
fractional maps acting on the  standard  $2$-simplex 
$\Delta = \{(x_1,x_2,x_3)\in\RR^3_+: x_1 + x_2 + x_3 = 1\}$. 
It is a fusion algorithm  such as introduced in \cite{EPTCS63.8,labbe_structure_2012} which 
combines the   two classical algorithms that are  Poincar\'e (\textbf{P})
algorithm and Arnoux-Rauzy (\textbf{AR}) algorithm, which are  respectively
defined  (under their linear form)  in dimension 3 as follows: Poincar\'e algorithm  acts on a triple of
non-negative entries by subtracting  the smallest entry to the median and the
median to the largest, whereas   Arnoux-Rauzy algorithm acts   by  subtracting
the  sum of the two smallest entries to the largest, when  possible.  Our
 algorithm privilegiates  an Arnoux-Rauzy  step if possible, otherwise
it perfoms a Poincar\'e step.

The  simplex $\Delta$    admits as  vertices  the vectors
$\be_1 = (1,0,0)^\top$,
$\be_2 = (0,1,0)^\top$ and
$\be_3 = (0,0,1)^\top$.
In order to partition $\Delta$, we consider
the following fifteen matrices, namely
\[
\footnotesize
\begin{array}{lllll}
A_{1} =
\left(\begin{array}{rrr}
1 & 1 & 1 \\
0 & 1 & 0 \\
0 & 0 & 1
\end{array}\right),&
P_{21} =
\left(\begin{array}{rrr}
1 & 1 & 1 \\
0 & 1 & 1 \\
0 & 0 & 1
\end{array}\right),&
P_{31} =
\left(\begin{array}{rrr}
1 & 1 & 1 \\
0 & 1 & 0 \\
0 & 1 & 1
\end{array}\right),&
H_{21} =
\left(\begin{array}{rrr}
1 & 0 & 0 \\
0 & 1 & 0 \\
1 & 0 & 1
\end{array}\right),&
H_{31} =
\left(\begin{array}{rrr}
1 & 0 & 0 \\
1 & 1 & 0 \\
0 & 0 & 1
\end{array}\right),
\\
A_{2} =
\left(\begin{array}{rrr}
1 & 0 & 0 \\
1 & 1 & 1 \\
0 & 0 & 1
\end{array}\right),&
P_{12} =
\left(\begin{array}{rrr}
1 & 0 & 1 \\
1 & 1 & 1 \\
0 & 0 & 1
\end{array}\right),&
P_{32} =
\left(\begin{array}{rrr}
1 & 0 & 0 \\
1 & 1 & 1 \\
1 & 0 & 1
\end{array}\right),&
H_{12} =
\left(\begin{array}{rrr}
1 & 0 & 0 \\
0 & 1 & 0 \\
0 & 1 & 1
\end{array}\right),&
H_{32} =
\left(\begin{array}{rrr}
1 & 1 & 0 \\
0 & 1 & 0 \\
0 & 0 & 1
\end{array}\right),
\\
A_{3} =
\left(\begin{array}{rrr}
1 & 0 & 0 \\
0 & 1 & 0 \\
1 & 1 & 1
\end{array}\right),&
P_{13} =
\left(\begin{array}{rrr}
1 & 1 & 0 \\
0 & 1 & 0 \\
1 & 1 & 1
\end{array}\right),&
P_{23} =
\left(\begin{array}{rrr}
1 & 0 & 0 \\
1 & 1 & 0 \\
1 & 1 & 1
\end{array}\right),&
H_{13} =
\left(\begin{array}{rrr}
1 & 0 & 0 \\
0 & 1 & 1 \\
0 & 0 & 1
\end{array}\right),&
H_{23} =
\left(\begin{array}{rrr}
1 & 0 & 1 \\
0 & 1 & 0 \\
0 & 0 & 1
\end{array}\right),
\end{array}
\]
whose column vectors  define    partitions  by triangles of the simplex such as  illustrated at 
Figure~\ref{fig:partition_arnoux_et_poincare} (left).
\begin{figure}[h]
\begin{center}
\begin{minipage}[c]{0.60\linewidth}
\includegraphics[width=0.30\linewidth]{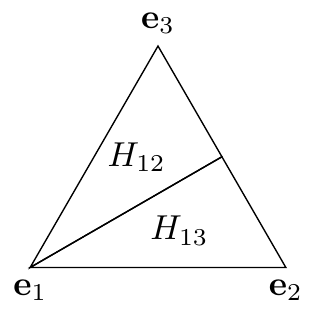}
\includegraphics[width=0.30\linewidth]{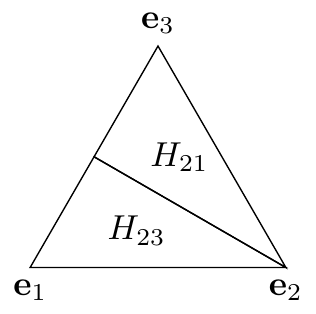}
\includegraphics[width=0.30\linewidth]{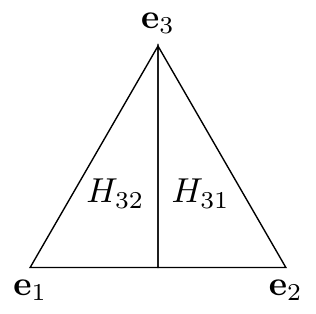}
\includegraphics[width=0.30\linewidth]{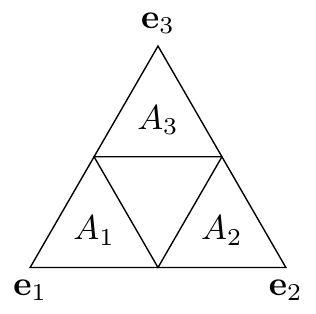}
\includegraphics[width=0.30\linewidth]{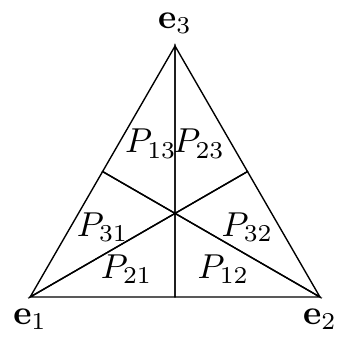}
\end{minipage}
\begin{minipage}[c]{0.37\linewidth}
\includegraphics[width=1.00\linewidth]{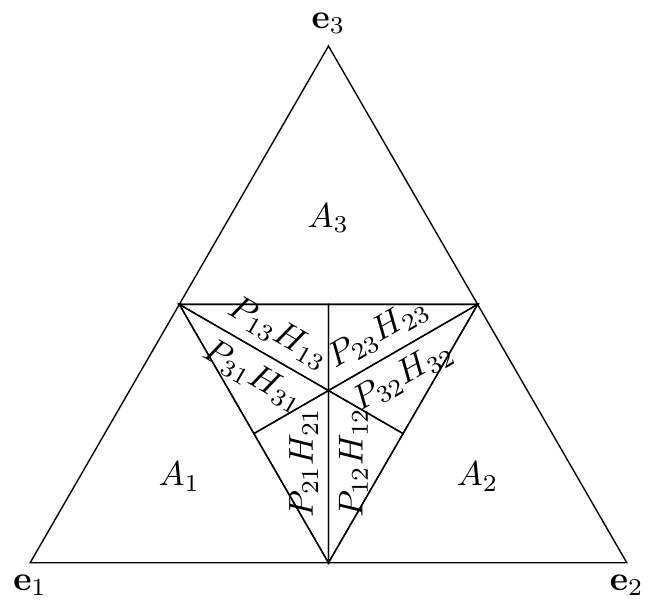}
\end{minipage}
\end{center}
\caption{Left: the partition provided  by  three Arnoux-Rauzy matrices,  the six Poincar\'e matrices and  the six half
triangles.  Right: the partition of Arnoux-Rauzy-Poincar\'e algorithm.}
\label{fig:partition_arnoux_et_poincare}
\end{figure}
Then, the column vectors of
$A_1$, $A_2$, $A_3$, $P_{31}H_{31}$, $P_{13}H_{13}$,
$P_{23}H_{23}$, $P_{32}H_{32}$, $P_{12}H_{12}$ and $P_{21}H_{21}$
describe a  partition of $\Delta$  depicted in
Figure~\ref{fig:partition_arnoux_et_poincare} (right). Partitions are
considered here up to a set of zero measure.
This partition allows one to associate with  almost every point  of $\Delta$ a
matrix as follows:
\[
\begin{array}{rcl}
    M:\Delta & \to & GL(3,\ZZ)\\
\bx & \mapsto &
\begin{cases}
A_k    & \text{ if } \bx \in A_k\Delta,\\
P_{jk}  & \text{ else if } \bx \in P_{jk}H_{jk}\Delta.
\end{cases}
\end{array}
\]
We say that $\bx=(x_1,x_2,x_3)\in\Delta$ is \emph{totally irrational} if
$x_1$, $x_2$, $x_3$ are linearly independent over $\QQ$. When  $\bx$ is not  a totally irrational vector,
 there might be    more than one choice for the matrix $M(\bx)$ in the previous definition. Nevertheless,  the matrix $M(\bx)$ is uniquely defined  for a  totally irrational vector. 

Then, the Arnoux-Rauzy-Poincar\'e algorithm is defined (by  renormalizing with respect to the simplex    $\Delta$) the linear map~$M$:
\[
\begin{array}{rcl}
T:\Delta & \to & \Delta\\
    \bx & \mapsto & 
    \displaystyle
    \frac{M(\bx)^{-1}\cdot \bx}{\left\Vert M(\bx)^{-1}\cdot \bx\right\Vert_1}  \, \cdot 
\end{array}
\]
Each  totally irrational  vector $\bx\in\Delta$ defines an orbit under the map $T$ and  a 
sequence of matrices $(M_n(\bx))_{n \in {\mathbb N}}$:
\[
    M_0(\bx)=\mathrm{Id},\quad
    M_n(\bx)=
    M(T^{n-1}(\bx))  \  \mbox{ for all }  n.
\]




\begin{example}\label{ex:part1}
Consider $\bx=(1,\pi,\sqrt{2})$. The first $5$ points of the orbit of $\bx$
under the map $T$ are
\[
\bx  \in A_2\Delta,\quad
T(\bx)\in P_{13}H_{13}\Delta,\quad
T^2(\bx)\in A_2\Delta,\quad
T^3(\bx)\in A_3\Delta,\quad
T^4(\bx)\in A_1\Delta, \cdots
\]
One has $M_0(\bx)=\mathrm{Id}$,
$M_1(\bx)=A_2$,
$M_2(\bx)=P_{13}$,
$M_3(\bx)=A_2$,
$M_4(\bx)=A_3$ and
$M_5(\bx)=A_1$.
\end{example}

\subsection{ Arnoux-Rauzy-Poincar\'e $S$-adic words}
We now  associate with the  Arnoux-Rauzy-Poincar\'e algorithm   a finite set
$\S$  of  substitutions as well as  $\S$-adic  words.

We first start with some terminology. We consider a   finite set  of {\em letters} ${\cal A}$, called   {\em alphabet}. Here ${\mathcal A}=\{1,2,3\}$.  A (finite)  {\em word} is an element of the free monoid ${\cal A}^*$ generated by ${\cal A}$.
The unique word of length $0$ is the \emph{empty word} and  we let it be  denoted as
$\emptyword$.
We let the set of all (finite) words over $\A$ be denoted  by $\A^*$. With the
concatenation of words as product operation, ${\mathcal A}^*$ is the free monoid with
$\emptyword$ as identity element.  A substitution on the alphabet ${\mathcal A}$  is a non-erasing morphism of the free monoid wich replaces letters  by  words.
Let $\sigma$ be a  substitution. Its {\em incidence matrix}  (also called  {\em abelianized matrix}) $M_{\sigma} = \left( m_{i,j} \right)_{1 \leq i,j \leq d}$ is defined 
as the square matrix  whose  entry of index $(i,j)$   is equal to  the number of occurrences  of the  letter $i$ in $\sigma(j)$.
If a word $u$ can be factorized as $pvs$, with $p,v,s\in\A^*$, then we say
that $p$ is a \emph{prefix}, $v$ is a \emph{factor} and $s$ is a \emph{suffix}
of $u$. The  factor $v$  is said \emph{proper} if  $p$ and $s$ are non-empty.  This notion extends to any infinite word $\bu$. The set ${\mathcal A} ^{\mathbb N}$  is  equipped
with the product topology of the discrete topology on each copy of ${\mathcal A}$; this topology is
induced by the following distance: for two distinct infinite words ${\bf u}$ and ${\bf v}$ in ${\mathcal A}^{\mathbb N}$, 
$\operatorname{d}({\bf u},{\bf v})=2^{-\min\{n \in { \mathbb N} \ \mid \ u_n\neq v_n\}}$. 

 The infinite word
$\bu \in {\mathcal A}^{\mathbb N}$  is said to admit an {\em  $S$-adic representation}  if there exist  a finite   set  $S$  of   substitutions   defined on the alphabet $\mathcal{A}$,
a sequence  $s = (\sigma_n)_{n\in\mathbb{N}} \in  S^{\mathbb N} $ 
of substitutions   that all belong to $ S$, and  $(a_n)_{n\in\mathbb{N}}$   a sequence of letters in ${\mathcal A}$ such that
$$\bu =  \lim_{n\to\infty} \sigma_0 \sigma_1 \cdots \sigma_n(a_n^{\infty}),$$
with the notation $a_n^{\infty}$ standing for the infinite  constant word taking the value $a_n$.
The word $\bu$ is  said to be $S$-adic, and the sequence $s$ is called the \emph{directive sequence}.  
We will use the  following notation: for all $m \in {\mathbb N}$
$$\bu  ^{(m)}=  \lim_{n\to\infty} \sigma_m \sigma_{m+1} \cdots \sigma_{n}(a_{n}^{\infty}).$$

An $S$-adic expansion with directive sequence $(\sigma_n)_{n \in {\mathbb N}}$ is said \emph{weakly primitive} if, for each $n$, there exists $r$ such that the substitution
$\sigma_{n} \cdots  \sigma_{{n+r}}$ is positive, that is,  its incidence matrix  has only positive entries.
If an infinite word $\bu$ admits a weakly primitive $S$-adic representation, then it is \emph{uniformly recurrent}, that is, all its factors
occur infinitely often and with bounded  gaps \cite{Durand:03}.  An infinite word $\bu$ is said  \emph{recurrent} if all its factors occur infinitely
often in $\bu$.
For more on $S$-adic
words, see \cite{BDRIMS,MR2759107,durand_do_2013,Leroy12}.


We now  associate   substitutions  with the  matrices  defining the Arnoux-Rauzy-Poincar\'e algorithm.
Let $i,j,k$ be 
such that $\{i,j,k\}=\{1,2,3\}$.
A {\em Poincar\'e substitution} is a substitution of   the form $\pi_{jk}: i\mapsto ijk, j\mapsto jk, k\mapsto k$.
An {\em Arnoux-Rauzy substitution} is given by $\alpha_{k}: i\mapsto ik, j\mapsto jk, k\mapsto k$.
For each $\{i,j,k\}=\{1,2,3\}$, $P_{jk}$ is the incidence matrix of the
substitution $\pi_{jk}$ and $A_k$ is the incidence matrix of $\alpha_k$.
There are thus  6 Poincar\'e and 3 distinct Arnoux-Rauzy substitutions:
\[
\begin{array}{lll}
\pi_{23} = \left\{\begin{array}{l}1 \mapsto 123 \\2 \mapsto 23 \\3 \mapsto 3
\\\end{array}\right., &
\pi_{13} = \left\{\begin{array}{l}1 \mapsto 13 \\2 \mapsto 213 \\3 \mapsto 3
\\\end{array}\right., &
\alpha_{3} = \left\{\begin{array}{l}1 \mapsto 13 \\2 \mapsto 23 \\3 \mapsto 3
\\\end{array}\right., \\
\pi_{12} = \left\{\begin{array}{l}1 \mapsto 12 \\2 \mapsto 2 \\3 \mapsto 312
\\\end{array}\right., &
\pi_{32} = \left\{\begin{array}{l}1 \mapsto 132 \\2 \mapsto 2 \\3 \mapsto 32
\\\end{array}\right., &
\alpha_{2} = \left\{\begin{array}{l}1 \mapsto 12 \\2 \mapsto 2 \\3 \mapsto 32
\\\end{array}\right., \\
\pi_{31} = \left\{\begin{array}{l}1 \mapsto 1 \\2 \mapsto 231 \\3 \mapsto 31
\\\end{array}\right., &
\pi_{21} = \left\{\begin{array}{l}1 \mapsto 1 \\2 \mapsto 21 \\3 \mapsto 321
\\\end{array}\right., &
\alpha_{1} = \left\{\begin{array}{l}1 \mapsto 1 \\2 \mapsto 21 \\3 \mapsto 31
\\\end{array}\right..
\end{array}
\]
Let 
\[
\S:=
\{\alpha_1, \alpha_2, \alpha_3, 
\pi_{12},
\pi_{13},
\pi_{21},
\pi_{23},
\pi_{31},
\pi_{32}\}.
\]
We also  denote by $\S_\alpha$, $\S_\pi$,  respectively, the following
sets of substitutions:
\[
\S_\alpha= \{\alpha_1, \alpha_2, \alpha_3\},\ 
\S_\pi   = \{\pi_{12}, \pi_{13}, \pi_{23}, \pi_{21}, \pi_{31}, \pi_{32}\}, \mbox{ with } 
\S      =  \S_\alpha\cup\S_\pi.
\]
The substitutions  in $\S$ are such that  for any letter $i \in \{1,2,3\}$, 
$\sigma(i)$ admits $i$ as a prefix. This yields the convergence of  any  $\S$-adic
representation in  $  {\mathcal A}^{\mathbb N}$ if the
sequence of letters $(a_n)_n$ is constant.  More precisely,
 for  any sequence  of  substitutions $(\sigma_n)_n$ with values in $\S$ and for  every  letter $a \in \{1,2,3\}$
then  the following  limit exists
$$ \lim_{n\to\infty} \sigma_0 \sigma_1 \cdots \sigma_n(a^{\infty}).$$

\begin{definition}[\bf  Arnoux-Rauzy-Poincar\'e $\S$-adic  word]\label{def:ARPSadicword}
An \emph{Arnoux-Rauzy-Poincar\'e $\S$-adic word} is an infinite word of the form 
$$\bu= \lim_{n\to\infty} \sigma_0 \sigma_1 \cdots \sigma_n(a^{\infty}),$$
where $a\in\A$ and $\sigma_n\in\S$ for all $n\geq0$.
Its directive sequence is the sequence $s=(\sigma_n)_n $.
\end{definition}

\subsection{The Arnoux-Rauzy-Poincar\'e $S$-adic system}
The aim of this section is to  associate with the  Arnoux-Rauzy-Poincar\'e
algorithm   an $\S$-adic  symbolic dynamical system by taking into account the
restrictions  provided by the algorithm which is not complete but Markovian.
We first  recall the definition of an $S$-adic system.  An  {\em $S$-adic
system} is  obtained  by  adding  restrictions on the set of allowed
directive sequences: it   is  given  by a finite directed strongly connected
graph $\G$ labeled by the substitutions, with each infinite path giving rise
to a directive sequence \cite{BDRIMS}.

The partition of $\Delta$ allows to associate with  almost any  point of $\Delta$ a
substitution of~$\S$:
\[
\begin{array}{rcl}
    \sigma:\Delta & \to & \S\\
\bx & \mapsto &
\begin{cases}
\alpha_k    & \text{ if } \bx \in A_k\Delta,\\
\pi_{jk}    & \text{ else if } \bx \in P_{jk}H_{jk}\Delta,
\end{cases}
\end{array}
\]
and a  directive sequence  $s=(\sigma_n)_n$ with $\sigma_n=
\sigma(T^n  (\bx))$ for all $n$.
Observe that the  substitution $\sigma(\bx)$ has for   incidence matrix  $M(\bx)$ such as defined in Section \ref{subsec:algo}.
\begin{definition}\label{def:ARPSadicwordbis}
An \emph{$\S$-adic word $\bu$ generated by the Arnoux-Rauzy-Poincar\'e
algorithm  applied to the  totally irrational vector $\bx \in \Delta$} is   an infinite word of the form
\[
    \bu = 
    \lim_{n\to\infty}
    \left(
    \sigma(\bx)\cdot
    \sigma(T(\bx))\cdot
    \sigma(T^2(\bx))\cdot
    \ldots\cdot
    \sigma(T^{n-1}(\bx))
\right)
    (a^{\infty})
\]
where $a\in\{1,2,3\}$.
Its directive sequence is the sequence $s=(\sigma_n)_n $ with  $\sigma_n=
\sigma(T^n  (\bx))$ for all $n$.
\end{definition}

Let us  show that the  factors of the  directive sequences  produced by the Arnoux-Rauzy-Poincar\'e
algorithm belong to a rational
language  strictly included in $\S^*$.
We consider the automaton $\G=(Q,\S,\delta,I,F)$  defined by
the states 
\[
Q=\{\Delta, H_{12}, H_{13},\\ H_{21}, H_{23}, H_{31}, H_{32}\},
\]
the alphabet $\S$, with
the transitions $\delta\subset Q\times\S\times Q$ being  defined by
\[
\begin{aligned}
\delta=\bigcup_{\{i,j,k\}=\{1,2,3\}} \{(\Delta,\alpha_k,\Delta),
&
(\Delta,\pi_{jk},H_{jk}),
(H_{jk},\alpha_{j},H_{jk}),\\
&
(H_{jk},\alpha_{i},\Delta),
(H_{jk},\pi_{ij},H_{ij}),
(H_{jk},\pi_{ki},H_{ki}),
(H_{jk},\pi_{ji},H_{ji})\},
\end{aligned}
\]
and with initial state $I=\{\Delta\}$
and final state $F=Q$ (see Figure~\ref{figure:markovchain}).
We consider the $\S$-adic system associated with the regular language $\L(\G)$.
\begin{figure}[h!]
\begin{center}
\includegraphics{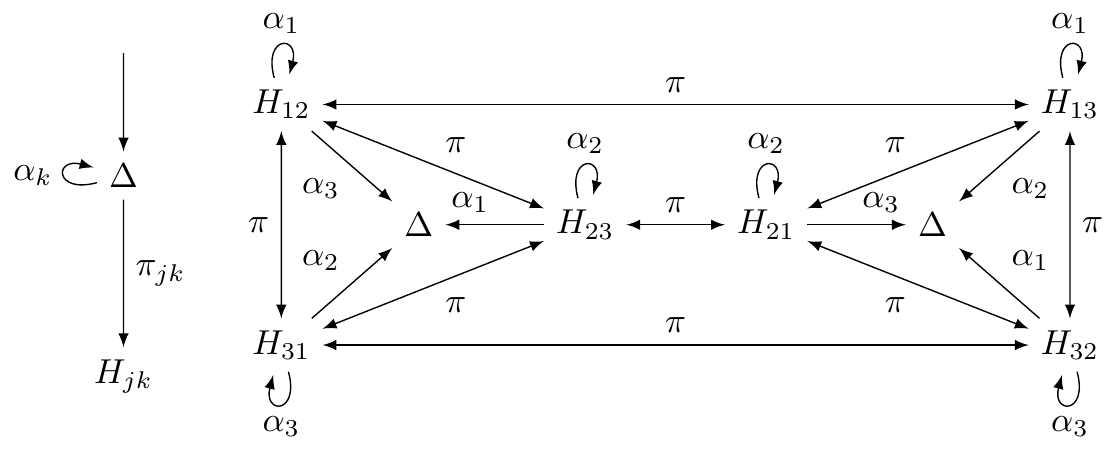}
\end{center}
\caption{The deterministic automaton $\G$. To avoid crossing arrows, the
initial state $\Delta$ is drawn at three places. The indices of $\pi$
transitions are not written since they are determined by the indices of the
arrival state: $\xrightarrow{\pi} H_{jk}$ means $\xrightarrow{\pi_{jk}}
H_{jk}$.}
\label{figure:markovchain}
\end{figure}
This language corresponds to directive sequences 
for which  the sequence of incidence matrices 
is generated by the execution of the Arnoux-Rauzy-Poincar\'e algorithm.
\begin{proposition}[\bf ARP regular language] \label{prop:1}
The  set of directive sequences produced by the
Arnoux-Rauzy-Poincar\'e algorithm is   included in   the set of  labeled infinite
paths in the automaton $\G$.
\end{proposition}
\noindent
The proof of the proposition is provided in the appendix.

\begin{remark}
We can even prove that the closure of  the  set of directive sequences produced by the
Arnoux-Rauzy-Poincar\'e algorithm is   equal to    the set  $X_{\mathcal G}$ of  labeled infinite
paths starting  in the automaton $\G$, as a consequence of the convergence of the algorithm proved in Section \ref{sec:convergence}.
 Let $\Sigma \colon  \Delta \rightarrow X_{\mathcal G}$ be the map 
 that associates with a   (totally irrational  vector) $\bx $  the directive
 sequence $(\sigma_n)_n$
where $\sigma_n  =\sigma (T^n (\bx))$ for all $n$.
 One has the following diagram and measure-theoretical isomorphism,  where $\Sigma$ is a.e.  one-to-one and where the shift  associates with  the label of an infinite path  the  label of the path  deprived of its
 first edge:
$$
\begin{array}{ccc}
\Delta&\stackrel{T}{\longrightarrow}& \Delta\\
\Big\downarrow  \scriptstyle{\Sigma} & &\Big\downarrow \scriptstyle{\Sigma} \\\
    X_{\mathcal G} &\underset{\mathrm{shift}}
{\longrightarrow}& X_{\mathcal G}
\end{array}
$$
\end{remark}
We now can define the Arnoux-Rauzy-Poincar\'e $\S$-adic  system  from the
multidimensional continued fraction algorithm.

\begin{definition}[\bf  Arnoux-Rauzy-Poincar\'e $\S$-adic system]\label{def:ARPSadic}
The \emph{Arnoux-Rauzy-Poincar\'e $\S$-adic system} is the set of   $\S$-adic 
words
$$\bu =  \lim_{n\to\infty} \sigma_0 \sigma_1 \cdots \sigma_n(a^{\infty}),$$
 whose directive sequence   $(\sigma_n)_n$ is an infinite path in  $\G$.
We distinguish three types of directive sequences together with some restrictions on the
chosen letter $a$:
\begin{enumerate}
\item if $(\sigma)_n\in   \S^*  \{\alpha_k\}^\NN$, for $k\in\{1,2,3\}$,  then $a=k$  (Type 1);
\item else if $(\sigma)_n \in   \S^* \{\alpha_k, \alpha_j\}^\NN$,
	then $a\in\{j,k\}$,  for some $\{i,j,k\}=\{1,2,3\}$   (Type 2);
\item otherwise,    take any  $a \in  \{1,2,3\}$ (Type 3).
\end{enumerate}
\end{definition}

The requirements in this definition concerning the choice of the letter $a$
will  be clearer with Proposition \ref{prop:recurrence} below: they aim at
working with  recurrent words which will be used in the  computation  of the
factor complexity function. 
According to Proposition \ref{prop:1},
any  $\S$-adic word $\bu$ generated by the Arnoux-Rauzy-Poincar\'e
algorithm  applied to a  totally irrational vector $\bx \in \Delta$, according to  Definition \ref{def:ARPSadicwordbis},
 belongs to the  Arnoux-Rauzy-Poincar\'e $\S$-adic system. Furthermore, they  correspond to Type 3 in Definition  \ref{def:ARPSadic}.

\begin{remark}
We  stress  the following terminology:  by Arnoux-Rauzy-Poincar\'e   $\S$-adic word,  we mean  an $\S$-adic word with  no other  restriction on  the  directive   sequence  that the fact  that it belongs to
$\S ^{\mathbb  N}$  (see Definition \ref{def:ARPSadicword}), whereas for  a  word    in the  Arnoux-Rauzy-Poincar\'e   $\S$-adic system, the restrictions of Proposition \ref{prop:1} are taken into account.
\end{remark}

\begin{example} We continue Example \ref{ex:part1}.
The word generated by the Arnoux-Rauzy-Poincar\'e algorithm  applied to 
$\bx=(1,\pi,\sqrt{2})$ is:
\begin{align*}
\bu
&=
\alpha_2\pi_{13}\alpha_2\alpha_3\alpha_1\pi_{31}\pi_{23}\pi_{31}\pi_{12}(\alpha_3)^8
\alpha_1(\alpha_2)^6\pi_{21}\alpha_3\alpha_3\alpha_1\pi_{32}
\cdots(1)\\
&= 1232212323221232212323221232123221232322123232212321232212323\cdots
\end{align*}
Note that the substitutions shown on the above line determine the prefix of
$\bu$ of length $1453060$.
The first prefixes are
\[
\alpha_2(1) = 12,\quad
\alpha_2\pi_{13}(1) = 1232,\quad
\alpha_2\pi_{13}\alpha_2(1) = 123221232,\quad
\alpha_2\pi_{13}\alpha_2\alpha_3(1) = 1232212323221232.
\]
Observe that due to its $S$-adic construction, the infinite  word $\bu$ can be decomposed  on    three-block  codes  (that is, on codes consisting   of  three    finite words)
in
many ways:
\begin{align*}
    &12322|1232322|12322|1232322|1232|12322|1232322|1232322|\cdots\\
    &12|32|2|12|32|32|2|12|32|2|12|32|32|2|12|32|12|\cdots\\
    &123|22123|23|22123|22123|23|22123|2123|22123|23|22123|\cdots
\end{align*}
The blocks are in each case respectively $\{12322, 1232322, 1232\}$,
$\{12, 2, 32\}$ and $\{23, 22123, 2123\}$ (they are obtained as    $\sigma_1 \cdots \sigma_n (i)$, for $i=1,2,3$, or else,  as    return words  on the letter $1$ in $\bu$,  where a  return word  on $1$ is   a finite word $v$   
that does not contain  the letter $1$, but  that  is such that     $v1$  is  a factor  of $\bu$).
For comparison, the billiard word of direction $(1,\pi,\sqrt{2})$ starting at $(0,0,0)$ is:
\[
2321232212322312232123221322231223212322321223212322132232123\cdots
\]
It    has quadratic factor complexity. It cannot be decomposed on a three-factor code   (its has too much  return words on  each letter). 
\end{example}

\subsection{Totally irrational vectors  and  weak primitivity}


The next lemma provides a characterization of  weakly primitive  $\S$-adic expansions.
Indeed weak primitivity fails if and only if the directive sequence $(\sigma_n)_{n \in {\mathbb N}}$ contains 
finitely  many  Poincar\'e substitutions and    takes ultimately  at most two values (that  thus are
  Arnoux-Rauzy substitutions). 

\begin{lemma}\label{lem:notweaklyprimitive}
Let $\bu=\lim_{n\to\infty}\sigma_0\sigma_1\cdots\sigma_n(a_n^{\infty})$ be an
$\S$-adic
word generated by the Arnoux-Rauzy-Poincar\'e algorithm applied to the vector
$\bx=(x_1,x_2,x_3)\in\Delta$. 
Its associated $\S$-adic expansion is  weakly primitive if and only if
\[
    (\sigma_n)_{n \in {\mathbb N}} \not \in \S^* \cdot 
    \left(
	\{\alpha_1, \alpha_2\}^\NN \cup
	\{\alpha_1, \alpha_3\}^\NN \cup
	\{\alpha_2, \alpha_3\}^\NN \cup
	\{\alpha_1\}^\NN \cup
	\{\alpha_2\}^\NN \cup
	\{\alpha_3\}^\NN
    \right).
\]
\end{lemma}

\begin{proof}
If    $(\sigma_n)_{n \in {\mathbb N}} \in \S^* \cdot 
    \left(
	\{\alpha_1, \alpha_2\}^\NN \cup
	\{\alpha_1, \alpha_3\}^\NN \cup
	\{\alpha_2, \alpha_3\}^\NN \cup
	\{\alpha_1\}^\NN \cup
	\{\alpha_2\}^\NN \cup
	\{\alpha_3\}^\NN
    \right)$, then  it is easily seen that $(\sigma_n)_{n \in {\mathbb N}}$ is not weakly primitive.

Now, let $(\sigma_n)_{n \in {\mathbb N}}$ be the
directive sequence 
of an $\S$-adic expansion which
is not weakly primitive in the Arnoux-Rauzy-Poincar\'e $\S$-adic system.
Being not weakly primitive means that there exists
$m$ such that for all $p$ with $m\leq p$ the substitution
$\sigma_{m} \cdots  \sigma_{{p}}$ is not positive, that is, one of the
entries of its incidence matrix is zero.
Moreover, for all $p$ and $r$ such that
$m\leq p\leq r$ the incidence matrix of the substitution
$\sigma_{p} \cdots  \sigma_{{r}}$ is not positive.

Note that since the incidence matrix of every substitution in $\S$ has entries $1$ on
the diagonal, the positivity of entries is preserved by left and right
multiplication. Therefore, if $\sigma_1\sigma_2\cdots\sigma_n\in\S^*$ is positive,
then $\varphi$ is positive for every
$\varphi\in\S^*\sigma_1\S^*\sigma_2\S^*\cdots\S^*\sigma_n\S^*$.

Assume first that  $(\sigma_n)_{n\geq m}$ contains no Poincaré substitution.
If $(\sigma_n)_{n\geq m}$
contains three distinct Arnoux-Rauzy substitutions, 
there are some values of $p$ and $r$
with $m\leq p \leq r$ such that
$\sigma_{p} \cdots  \sigma_{{r}}$  
contains three distinct Arnoux-Rauzy substitutions.
One  verifies that $\alpha_i\alpha_j\alpha_k$ is positive for all possible values
of $i$, $j$, $k$ with $\{i,j,k\}=\{1,2,3\}$.
Then
$\sigma_{p} \cdots  \sigma_{{r}}$ is positive  which is a contradiction.
Therefore, we conclude that
$(\sigma_n)_{n \in {\mathbb N}} \in \S^* \cdot 
    \left(
	\{\alpha_1, \alpha_2\}^\NN \cup
	\{\alpha_1, \alpha_3\}^\NN \cup
	\{\alpha_2, \alpha_3\}^\NN
    \right)$.

Assume $(\sigma_n)_{n\geq m}$ contains at least one Poincaré substitution. We
may suppose
that $(\sigma_{n})_{n\geq p}$ starts with a Poincaré substitution
$\sigma_p=\pi_{jk}$ for $p\geq m$.
Since $(\sigma_{n})_{n\geq p}\in\L(\G)$, then
\[
    (\sigma_n)_{n\geq p} \in
    \left(\pi_{jk}\alpha_j^\infty\right)
    \cup
    \left(\pi_{jk}\alpha_j^t\{\alpha_i,\pi_{ki},\pi_{ji}\}\S^\NN\right)
    \cup
    \left(\pi_{jk}\alpha_j^t\pi_{ij}\alpha_i^\infty\right)
    \cup
    \left(\pi_{jk}\alpha_j^t\pi_{ij}\alpha_i^s\{\alpha_k,\pi_{jk},\pi_{ik},\pi_{ki}\}\S^\NN\right),
\]
for some non-negative integers $s$ and $t$ and $\{i,j,k\}=\{1,2,3\}$.
But $\pi_{jk}\alpha_{i}$, $\pi_{jk}\pi_{ki}$ and $\pi_{jk}\pi_{ji}$ are positive.
Also
$\pi_{jk}\pi_{ij}\alpha_{k}$,
$\pi_{jk}\pi_{ij}\pi_{jk}$,
$\pi_{jk}\pi_{ij}\pi_{ik}$ and
$\pi_{jk}\pi_{ij}\pi_{ki}$ are positive.
Therefore,
\[
    (\sigma_n)_{n\geq p} \in
    \left(\pi_{jk}\alpha_j^\infty\right)
    \cup
    \left(\pi_{jk}\alpha_j^t\pi_{ij}\alpha_i^\infty\right)
\]
and we have shown that 
$(\sigma_n)_{n \in {\mathbb N}} \in \S^* \cdot 
    \left(
	\{\alpha_1\}^\NN \cup
	\{\alpha_2\}^\NN \cup
	\{\alpha_3\}^\NN
    \right)$.
\end{proof}

\begin{proposition}\label{prop:primitive}
Let $\bu$ be an
$\S$-adic
word generated by the Arnoux-Rauzy-Poincar\'e algorithm applied to  the   totally irrational vector
$\bx\in\Delta$.  Then 
the associated $S$-adic expansion is weakly  primitive.
In particular,  $\bu^{(m)}$ is of Type $3$, uniformly recurrent  and proper, for all $m$.
\end{proposition}

\begin{proof}
The conclusion follows from Lemma~\ref{lem:notweaklyprimitive} by noticing that if 
\[
    (\sigma_n)_{n \in {\mathbb N}} \in \S^* \cdot 
    \left(
	\{\alpha_1, \alpha_2\}^\NN \cup
	\{\alpha_1, \alpha_3\}^\NN \cup
	\{\alpha_2, \alpha_3\}^\NN \cup
	\{\alpha_1\}^\NN \cup
	\{\alpha_2\}^\NN \cup
	\{\alpha_3\}^\NN
    \right)
\]
then $\bx$ cannot be totally irrational.
\end{proof}

Observe that  not
every word of the  Arnoux-Rauzy-Poincar\'e $\S$-adic system is  uniformly recurrent. Nevertheless, one
easily checks that   words of this system are  all recurrent.

\begin{proposition}\label{prop:recurrence}
Any  infinite word   ${\bf u}$ in the  Arnoux-Rauzy-Poincar\'e system  is  recurrent 
as well as $\bu ^{(m)}$ for any $m$.

\end{proposition}
\begin{example}
The infinite word $\alpha_1^{\infty} (2^{\infty})=21^{\infty} $ is not recurrent  whereas   $\alpha_1^{\infty} (1^{\infty})=1^{\infty} $ is recurrent.
\end{example}

The restriction of the infinite words under study to the case where each letter always appears
as proper factor will also  be useful to prove the main result of this article.

\begin{definition}[\bf Proper  word ]\label{def:proper}
A word  $\bu \in \{1,2,3\}^{\mathbb N}$
 is  said \emph{proper} if 
each letter
$i\in\{1,2,3\}$ is a proper factor of  $\bu$, or equivalently,  for each letter  $i\in\{1,2,3\}$, 
there exists a letter $e$ such that $ei$ is a factor of $\bu$.
\end{definition}
\section{Factor complexity}\label{sec:prelim}

In this section, we define the terminology relative to languages, bispecial
factors, extension types and factor complexity. We adopt the notation of \cite{MR2759107}.

\subsection{Language and complexity function $p(n)$}

Let $\A=\{1,2,\ldots,d\}$ be an alphabet. 
 The length of a word $u\in\A^n$ is denoted  by $|u|$ and is equal to $n$,  whereas the notation $|u|_i$ stands for the number of   occurrences of the letter $i$ in $u$.
A \emph{language} is a subset of the free monoid $\A^*$. A
language $L$ is \emph{factorial} if  for any  $w\in L$, then any   factor  $u$ of $w$ belongs to $L$.
The abelianized of a finite word $w\in\A^*$ is the vector
\[
\vect{w} = (|w|_1, |w|_2, \ldots, |w|_d) \in \NN^d.
\]
We consider an infinite word $\bu=u_0u_1u_2u_3\cdots\in\A^\NN$.
For each $n\in\NN$, $\L_n(\bu)$ is the set of factors of length $n$ in $\bu$,
while $\L(\bu)$ is the set of all factors in $\bu$, and is called the
\emph{language of $\bu$}.
The language of $\bu$ is factorial.
For each $n\in\NN$, let $p_\bu(n)$ be the cardinality of $\L_n(\bu)$. Then
$p_\bu:\NN\to\NN$ is a function called the \emph{ factor complexity function} of $\bu$.
When no confusion is possible, we omit $\bu$ and just write $p$.

\subsection{Bispecial Factors and Extension Types}\label{sec:bispecial}

Let $w$ be a factor of  either a recurrent infinite word  or of a finite word $\bu$.   We let $E^+(w) =
\{x\in \A \mid wx \in \L(\bu)\}$ denote  the  set of right extensions of $w$ in
$\bu$.
The \emph{right valence} $d^+(w) = \Card\,E^ +(w)$ of $w$ (in $\bu$)   is defined as the
number of distinct right extensions of $w$. 
\emph{Left extensions} $E^-(w)$ and \emph{left valence} $d^-(w)$ are defined
in a  a similar way. A factor whose right valence is at least $2$ is called
\emph{right special}. A factor whose left valence is at least $2$ is called
\emph{left special}. A factor which is both left and right special is called
\emph{bispecial}. 
 The \emph{extension type} $E_\bu(w)$ of a factor $w$ of $\bu$ is
the set of pairs $(a,b)$ of $\A\times \A$ such that $w$ can be extended in
both directions as $awb$:
\[
E_\bu(w) = \{(a,b)\in \A\times \A  \mid awb\in \L(\bu)\}.
\]
We also use the notation $E_\bu(w)$ by $E(w)$ when the context is clear.
The \emph{bilateral multiplicity} of a factor $w$ is
the number
\[
m(w) = \Card \, E(w) - d^-(w) - d^+(w) + 1.
\]
We have the following  fact  (see e.g. \cite[Proposition 4.5.1]{MR2759107}) which links bilateral multiplicity to the notion
of bispecial factor: let $w$ be a factor of    a recurrent infinite word  such that $m(w)\neq 0$;
then, $w$ is bispecial.
A bispecial factor is said \emph{strong} if $m(w)>0$,
\emph{weak} if
$m(w)<0$ and \emph{neutral} if $m(w)=0$.
A bispecial factor is \emph{ordinary} if there exist letters $a,b\in\A$ such that
\begin{equation}\label{eq:defordinaire}
\{(a,b)\}
\subseteq
E(w)
\subseteq
\left(\{a\}\times \A\right) \cup \left(\A\times \{b\}\right).
\end{equation}
An ordinary bispecial factor is
neutral, but the converse is not true for $|\A|>2$.
We will  use this notion in particular in Section  \ref{subsec:life}.
\begin{lemma}
If a bispecial factor is ordinary, then it is neutral.
\end{lemma}
\begin{proof}
If $w$ is ordinary, then
\[
\Card \, E(w) =
\Card \, E^-(w) +
\Card \, E^+(w) - 1
\]
because $(a,b)\in E(w)$.
Thus, we have
$m(w) = \Card \, E(w) - d^-(w) - d^+(w) + 1=0$.
\end{proof}
It is convenient to represent  the extension type $E(w)$ of a bispecial factor $w$ in a graphical 
way. It is often represented as a bipartite graph, but  we choose here a table
representation: a cross ($\times$) is drawn at the intersection of row $a$ and
column $b$ if and only if $(a,b)\in E(w)$ (see Figure~\ref{fig:bispeciaux}).
\begin{figure}[h!]
\begin{center}
\begin{minipage}[c]{0.19\linewidth}
\centering
\footnotesize
{\scriptsize
\begin{tabular}{c|ccc}
  & 1 & 2 & 3 \\
\hline
1 &   & $\times$ &   \\
2 &   & $\times$ &   \\
3 & $\times$ & $\times$ & $\times$
\end{tabular}}\\
$m(w)=0$\\
neutral and ordinary
\end{minipage}
\begin{minipage}[c]{0.19\linewidth}
\centering
\footnotesize
{\scriptsize
\begin{tabular}{c|ccc}
  & 1 & 2 & 3 \\
\hline
1 &   & $\times$ &   \\
2 &   &   & $\times$ \\
3 & $\times$ & $\times$ & $\times$
\end{tabular}}\\
$m(w)=0$\\
neutral but not ordinary
\end{minipage}
\begin{minipage}[c]{0.19\linewidth}
\centering
\footnotesize
{\scriptsize
\begin{tabular}{c|ccc}
  & 1 & 2 & 3 \\
\hline
1 &   & $\times$ &   \\
2 &   &   &   \\
3 &   &   & $\times$
\end{tabular}}\\
$m(w)=-1$\\
weak
\end{minipage}
\begin{minipage}[c]{0.19\linewidth}
\centering
\footnotesize
{\scriptsize
\begin{tabular}{c|ccc}
  & 1 & 2 & 3 \\
\hline
1 &   &   &   \\
2 &   & $\times$ & $\times$ \\
3 & $\times$ & $\times$ & $\times$
\end{tabular}}\\
$m(w)=1$\\
strong
\end{minipage}
\end{center}
\caption{Examples of   tables representing  the  extension type $E(w)$ of a bispecial factor $w$.}
\label{fig:bispeciaux}
\end{figure}


\begin{definition}[\bf Left equivalence]\label{def:leftequivalent}
Let $w$ and $w'$ be two bispecial factors defined on the alphabet $\A$. We say that their extension types
are \emph{left equivalent} if there exists a permutation $\tau$ acting on $\A$ such that $E(w') = \{(\tau(a),b) \mid  (a,b) \in E(w) \}$.
\end{definition}
Right equivalence is defined similarly.
Left equivalence can be interpreted on the  table representation of
the extension type as follows. Indeed one representation  can be obtained from the other by a permutation of
the rows:
\[
\scriptsize
E(w) =
\begin{array}{c|ccc}
  & 1      & 2      & 3      \\
\hline
1 &        &        & \times \\
2 &        &        &        \\
3 & \times & \times & \times
\end{array}
\quad \quad \quad
\text{\normalsize and}
\quad \quad \quad
E(w') =
\begin{array}{c|ccc}
  & 1      & 2      & 3      \\
\hline
1 & \times & \times & \times \\
2 &        &        & \times \\
3 &        &        &
\end{array}
\]
Substitutions considered in this article preserve the first letter and thus
preserve the right extensions. Then, the notion of left equivalence is
sufficient for our need. But in general, we have the following definition.
Of course if the extension type of $w$ and $w'$ are left or right equivalent,
then they are also equivalent.
When the extension type of two words are equivalent, they share common
properties. In particular, being ordinary, strong or weak is preserved under
equivalence.
\begin{lemma}\label{lem:equivalence}
Let $w$ and $w'$ be two bispecial factors such that the
extension type of $w$ and $w'$ are equivalent, then
\begin{enumerate}[label=\textbullet]
\item $w$ is ordinary (neutral, strong, weak resp.) if and only if $w'$ is
    ordinary (neutral, strong, weak resp.),
\item $\Card E(w)=\Card E(w')$, $d^-(w)=d^-(w')$, $d^+(w)=d^+(w')$, $m(w)=m(w')$,
\item if the extension type of $w$ and $w'$ are left equivalent, then
$E^+(w)=E^+(w')$,
\item if the extension type of $w$ and $w'$ are right equivalent, then
$E^-(w)=E^-(w')$.
\end{enumerate}
\end{lemma}

\subsection{Factor Complexity}

Let $p(n)$ be the factor complexity function of the infinite word $\bu$.
Two other functions derived from the factor complexity are useful, namely 
the sequences of \emph{finite differences of order $1$ and $2$} respectively
of $p(n)$:
\begin{eqnarray}
    s(n)&=&p(n+1)-p(n),\label{eq:defsn}\\
    b(n)&=&s(n+1)-s(n)\label{eq:defbn}.
\end{eqnarray}
Of course, we have
\begin{eqnarray}
p(n) &=& p(0) + \sum_{\ell=0}^{n-1} s(\ell),\label{eq:pn}\\
s(n) &=& s(0) + \sum_{\ell=0}^{n-1} b(\ell).\label{eq:sn}
\end{eqnarray}
These equations are very useful to compute the complexity function $p(n)$ when
its growth is slow (for example  in the case of a linear growth), since in this case functions
$s$ and $b$ take small values.  For example, we have
$p(n)=n+1$ for all $n$  ($\bu$  is thus a Sturmian word) if and only if $s(n)$ is always equal to $1$, which is  also equivalent  to the fact that
exactly two letters occur ($p(1)=2$, $s(0)=1$) and  that $b(n)$ always takes the value $0$.

In this article, one of our main results is to show that some infinite words on
a three-letter alphabet have complexity $p(n)<3n$. In order to achieve this,
we use the next  lemma.

\begin{lemma}\label{lem:pn23iffsumbn01}
Suppose $|\A| = 3$. Then,
$p(n+1)-p(n)\in\{2,3\}$ if and only if
$\sum_{\ell=0}^{n-1} b(\ell) \in \{0,1\}$.
Furthermore, if the sequence of finite differences of order $2$ is such that
\[
(b(\ell))_\ell = 0,\ldots,0,1,0,\ldots,0,-1,0,\ldots,0,1,0,\ldots,0,-1,\ldots
\]
then
$\sum_{\ell=0}^{n-1} b(\ell) \in \{0,1\}$.
\end{lemma}
\begin{proof}
Since $|\A| = 3$, then $p(1)=3$ and $s(0)=p(1)-p(0)=3-1=2$.
We have
\[
p(n+1)-p(n)= s(n)
= s(0) + \sum_{\ell=0}^{n-1} b(\ell)
= 2    + \sum_{\ell=0}^{n-1} b(\ell),
\]
which yields the proof  of the  first statement.

The  proof of the second one comes from the fact that the first non-zero term of the sequence
$(b(\ell))_\ell$ is $+1$.
\end{proof}



The finite differences of order $1$ and $2$ of $p(n)$ are related to
special and bispecial factors as explained in \cite{MR1440670}. We state a
weaker form (for recurrent words) of a result of \cite{MR2759107}. Indeed,
as we are interested in the factor complexity of some recurrent
words, we do not need to consider unioccurrent or exceptional prefixes.

\begin{theorem}{\rm\cite[Theorem 4.5.4]{MR2759107}}\label{thm:cassaigne454}
Let $\bu\in \A^\NN$ be an infinite recurrent word.
Then, for all $n\in\NN$:
\begin{eqnarray}
s(n) &=& \sum_{ w \in {\mathcal L}_n (u)} (d^+(w) - 1)=
\sum_{ w \in {\mathcal L}_n (u)} (d^-(w) - 1)\\
b(n) &=& \sum_{ w \in {\mathcal L}_n (u)} m(w).\label{eq:bisp_mult}
\end{eqnarray}
\end{theorem}

\section{Bispecial Factors under Arnoux-Rauzy and Poincar\'e
Substitutions}\label{sec:arpbispecials}

The goal of the next sections is to describe  
factors  of   Arnoux-Rauzy-Poincar\'e $\S$-adic words.
The key   ingredient  is  a synchronization lemma  that allows the  desubstitution   with  respect to the substitutions
in $\S$ (Section \ref{lem:synch}). As  a consequence  for  bispecial  factors,  antecedents (they are uniquely defined and  always  bispecial)  and  bispecial images     together with their possible extensions   are  described  in details 
 in Section \ref{subsec:AR} and \ref{subsec:P}, respectively.
We then can consider the  notion of life of   a bispecial factor produced  by an $\S$-adic expansion (Section  \ref{subsec:life}).
So far we still  do not use  the  restrictions  of Proposition \ref{prop:1} on the  possible  directive sequences in $\S ^{\mathbb N}$ (they will be considered only in Section \ref{sec:proofthm}).
Section \ref{sec:quadratic} illustrates the fact that a  quadratic factor complexity can be reached without these restrictions.
We end this section with    the introduction of  notions of   order on vectors  allowing  the comparison of   abelianized  vectors  under the application
of substitutions in   $\S$ (Section \ref{sec:partialorder}).

 We recall that a Poincar\'e substitution is of the form $\pi_{jk}: i\mapsto ijk, j\mapsto jk, k\mapsto k$.
An Arnoux-Rauzy substitution is given by $\alpha_{k}: i\mapsto ik, j\mapsto jk, k\mapsto k$.

\subsection{Synchronization lemma}\label{lem:synch}

From now on, the alphabet is set to $\A=\{1,2,3\}$.
The next lemma describes the preimage of a factor under Arnoux-Rauzy ({\bf
AR}) and Poincar\'e ({\bf P}) substitutions.
Such statements are classical tools when computing the factor complexity  of fixed points  of substitutions.

\begin{lemma}[\bf Synchronization]\label{lem:preimagearnouxETpoincare}
Let $u\in\A^*$ and $w$ be a factor of $\alpha_k(u)$
for some $\{i,j,k\}=\{1,2,3\}$.
\begin{enumerate}[\rm (i)]
\item If $w$ is empty or if the first letter of $w$ is $i$ or $j$, then there
exist  a unique $v\in\A^*$ and a unique $s\in\{\emptyword,i,j\}$ such that $
w = \alpha_k(v)\cdot s$.
\item If the first letter of $w$ is $k$, then there
exist a unique $v\in\A^*$ and a unique $s\in\{\emptyword,i,j\}$ such that
$
w = k\cdot \alpha_k(v)\cdot s.$
\end{enumerate}
Let $u\in\A^*$ and $w$ be a factor of $\pi_{jk}(u)$
for some $\{i,j,k\}=\{1,2,3\}$.
\begin{enumerate}[\rm (i),resume]
\item If $w$ is empty or if the first letter of $w$ is $i$, then there
exist a unique $v\in\A^*$ and a unique $s\in\{\emptyword,i,j,ij\}$ such that
$
w = \pi_{jk}(v)\cdot s$.
\item If $w=j$, then there
exist a unique $v(=\emptyword)$ such that
$
w = j\cdot \pi_{jk}(v)$.
\item If the first letter of $w$ is $j$ and $|w|>1$, then there
exist a unique $v\in\A^*$ and a unique $s\in\{\emptyword,i,j,ij\}$ such that
$
w = jk\cdot \pi_{jk}(v)\cdot s$.
\item If the first letter of $w$ is $k$, then there
exist  a unique $v\in\A^*$ and a unique $s\in\{\emptyword,i,j,ij\}$ such that
$
w = k\cdot \pi_{jk}(v)\cdot s$.
\end{enumerate}
\end{lemma}

\begin{proof}
The sets   $\{ik, jk, k\}$   and $\{ijk, jk, k\}$ form a prefix code.
\end{proof}

\begin{definition}[\bf Antecedent, extended image]
Let $\sigma=\alpha_k$ or $\sigma=\pi_{jk}$, $u\in\A^*$ and $w$ be a
factor of $\sigma(u)$. We say that the \emph{antecedent of $w$ under $\sigma$}
is the unique word $v$ as defined by Lemma~\ref{lem:preimagearnouxETpoincare}.
If $v$ is the antecedent of a word $w$, then we say that the word $w$ is an
\emph{extended image} of $v$.
\end{definition}

Note that  the antecedent is unique,  but  that  a word $v$ may have more than one extended
image. Consider for instance $w_1=23 \pi_{23}(11)1=231231231$ and $w_2=3
\pi_{23}(11)2=31231232$ which  are two distinct extended images of $v=11$.  This is
why the situation becomes here quite intricate    especiallly for bispecial
factors. In fact, it happens that strong and weak bispecial words  appear in pairs:  the image of a neutral bispecial factor $v$ can have two extended images that
are bipsecial,  with one of them being strong, and the other one  being  weak.  For more details, see 
Lemma \ref{lem:lifebispecial} and Remark \ref{rem:pairs} below.

We now  consider images and antecedents of bispecial factors.
\begin{definition}[\bf Bispecial extended image]
Let $u\in\A^*\cup\A^\NN$ and $v$ be a
factor of $u$. 
We shall say that a \emph{bispecial extended image} $w$ of $v$ under $\sigma$
is a bispecial word     of $\sigma(u)$ which is an extended image of $v$ under~$\sigma$.
\end{definition}

For example, let $v$ be a bispecial factor and suppose
$E(v)=\{(1,2),(2,3),(3,1), (3,2),(3,3)\}$.  Then $w=3\pi_{23}(v)$ and
$w'=23\pi_{23}(v)$ are both bispecial extended images of $v$ under $\pi_{23}$.
Indeed, we have
\[
    \pi_{23}\left(\{1v2,2v3,3v1,3v2,3v3\}\right)
    =
    \{123\pi_{23}(v)23,23\pi_{23}(v)3,3\pi_{23}(v)123,3\pi_{23}(v)23,3\pi_{23}(v)3\}
\]
and the extension types are
$E(w)=\{(2,2),(2,3),(3,1),(3,2),(3,3)\}$ and
$E(w')=\{(1,2),(3,3)\}$.

The next lemma allows one to relate every bispecial factor to a shorter one and
eventually to the empty word.

\begin{lemma}[\bf Bispecial extended image growth] \label{lem:growth}
Let $\sigma=\alpha_k$ or $\sigma=\pi_{jk}$ and $w\neq\emptyword$ be a 
non-empty bispecial extended image of $v$ under $\sigma$. Then, $|v|<|w|$.
\end{lemma}
\begin{proof}
Suppose that $\sigma=\alpha_k$ for some $k\in\{1,2,3\}$.
Since $w$ is non-empty, $w$ starts and ends with letter $k$ and from
Lemma~\ref{lem:preimagearnouxETpoincare} (ii), the unique antecedent $v$ of $w$ is such
that $w=k\alpha_k(v)$. We conclude that $|v|<|w|$.

Suppose that $\sigma=\pi_{jk}$ for some $\{i,j,k\}=\{1,2,3\}$.  Since $w$ is
non-empty, $w$ starts with letter $j$ or $k$ and ends with letter $k$.  From
Lemma~\ref{lem:preimagearnouxETpoincare} (iv) and (v), the unique antecedent $v$ of
$w$ is such that $w=k\pi_{jk}(v)$ or $w=jk\pi_{jk}(v)$.  In both cases,
$|v|<|w|$.
\end{proof}

\subsection{Arnoux-Rauzy substitutions}\label{subsec:AR}

The case  of Arnoux-Rauzy substitutions is particularly convenient to handle,  both  
for  bispecial  extended  images   or  for  antecedents of bispecial factors.

\begin{lemma}[\bf AR - Bispecial extended image]\label{lem:extendedimagesAR}
Let $u\in\A^*$  and let $v$ be a bispecial factor of  $u$.
There is a unique bispecial extended image $w=k\alpha_k(v)$ of $v$  in  $\alpha_k (u)$.
\end{lemma} 

\begin{proof}
Let $w$ and $w'$ be two extended images of $v$ under $\alpha_k$. Since they
are bispecial factors,  one deduces from   Lemma~\ref{lem:preimagearnouxETpoincare} that both $w$ and $w'$ start and end with letter $k$.
Hence  $w=k\alpha_k(v)=w'$.
\end{proof}


\begin{lemma}[\bf AR - Antecedent of a bispecial]\label{lem:arbispecial}
Let $u\in\{1,2,3\}^*$ and $w\neq\emptyword$ be a bispecial factor of
$\alpha_k(u)$.
Let $v$ be the unique antecedent of $w$ under $\alpha_k$. One has 
 $w=k\alpha_k(v)$.
Furthermore, $v$ is bispecial and it has the same extension type
$E_{\alpha_k(u)}(w)=E_{u}(v)$ and same multiplicity $m(w)=m(v)$ as $w$.
\end{lemma}
\begin{figure}[h!]
\begin{center}
\includegraphics{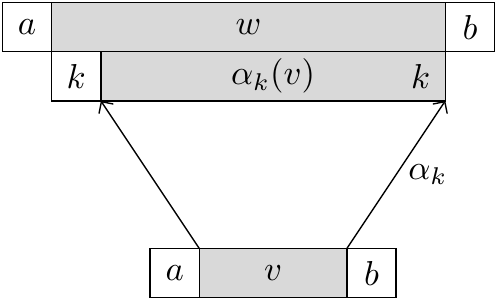}
\end{center}
\caption{The preimage of the bispecial word $w$ under $\alpha_k$.}
\label{fig:preimage_arnoux}
\end{figure}
\begin{proof}
One checks that $(a,b)\in E(v)$ if and only if
$(a,b)\in E(k\alpha_k(v))$ (see Figure~\ref{fig:preimage_arnoux}).
Then $E(k\alpha_k(v))=E(v)$. We deduce that $E^+(k\alpha_k(v))=E^+(v)$ and
$E^-(k\alpha_k(v))=E^-(v)$.  From this we conclude that $m(k\alpha_k(v))=m(v)$.
\end{proof}

\subsection{Poincar\'e substitutions}\label{subsec:P}
The  case of  Poincar\'e substitutions is  more  delicate to handle as already illustrated by the following  result.
We loose here unicity  for the bispecial extended images.

\begin{lemma}[\bf P - Bispecial extended images]\label{lem:extendedimagesP}
Let  $i,j,k$ such that $\{i,j,k\}=\{1,2,3\}$.  Let $u\in\{1,2,3\}^*$  and let $v$ be a bispecial factor of  $u$.
There are at most two distinct bispecial extended images of $v$ under
$\pi_{jk}$. They are either $k\pi_{jk}(v)$ or $jk\pi_{jk}(v)$.
\end{lemma}

\begin{proof}
Let $w$ be a bispecial extended image of $v$ under $\pi_{jk}$.
Since $w$ is a bispecial factor, it must start with letter $j$ or
$k$ and end with letter $k$.  From Lemma~\ref{lem:preimagearnouxETpoincare},  one gets $w \in
\{jk\pi_{jk}(v), k\pi_{jk}(v)\}$.
\end{proof}

The ``at most two" of Lemma~\ref{lem:extendedimagesP} will be made more  precise later
in Lemma~\ref{lem:detailedpoincare} where conditions will be given for when a
bispecial factor  has one or two bispecial extended images under a Poincar\'e
substitution.

In order to get a similar result concerning the  antecedent of  a bispecial factor
under  Poincar\'e substitutions (see Lemma~\ref{lem:pbispecial} below), we first need the
following result stated for factors in general which is also used for proving
Lemma~\ref{lem:detailedpoincare} and \ref{lem:lifebispecial}.

\begin{lemma}[\bf P - Extensions]\label{lem:poincareextensions}
Let  $i,j,k$ such that $\{i,j,k\}=\{1,2,3\}$. Let $u\in\{1,2,3\}^*$ and $v$ be a factor of $u$. We assume that  for all $(a,b) \in E(v)$, there exists a letter
 $e$ such that $eavb$ is also  a factor of $u$. 
The extensions of $v$ in $u$ are related to the extensions of
$k\pi_{jk}(v)$ and $jk\pi_{jk}(v)$ considered as factors of
$\pi_{jk}(u)$:
\[
\begin{array}{l}
(i,b) \in E(v) \iff (j,b) \in E(k\pi_{jk}(v))
\quad\text{and}\quad(i,b) \in E(jk\pi_{jk}(v)),\\
(j,b) \in E(v) \iff (j,b) \in E(k\pi_{jk}(v))
\quad\text{and}\quad(k,b) \in E(jk\pi_{jk}(v)),\\
(k,b) \in E(v) \iff (k,b) \in E(k\pi_{jk}(v)).
\end{array}
\]
\end{lemma}

\begin{proof}
First note that $i\notin E^-(k\pi_{jk}(v))$ and $j\notin E^-(jk\pi_{jk}(v))$.
Note also that the right extensions are preserved by $\pi_{jk}$ because $\pi_{jk}$
preserves the first letter of words.
Let $(a_0,b)\in E(v)$, $(a_1,b)\in E(k\pi_{jk}(v))$ and $(a_2,b)\in
E(jk\pi_{jk}(v))$ and let us consider each case $a_0=i$, $a_0=j$ and $a_0=k$
separately (see Figure~\ref{fig:preimage_poincare}).
\begin{figure}[h!]
\begin{center}
\includegraphics{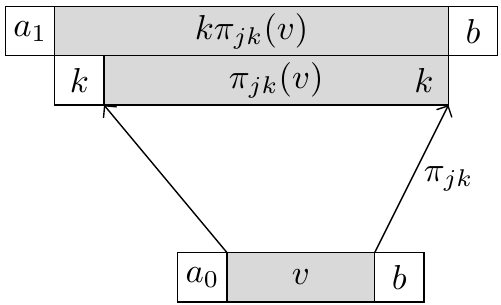}
\quad\quad\quad\quad
\includegraphics{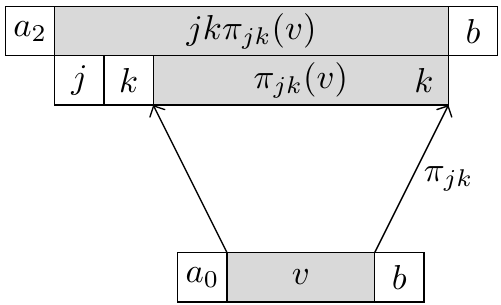}
\end{center}
\caption{The preimage of $k\pi_{jk}(v)$ and $jk\pi_{jk}(v)$ under $\pi_{jk}$.}
\label{fig:preimage_poincare}
\end{figure}
According to the assumption made on $v$, one  checks that
if $a_0=i$, then $a_1=j$ and $a_2=i$;
if $a_0=j$, then $a_1=j$ and $a_2=k$;
if $a_0=k$, then $a_1=k$.
The reciprocals are also verified.
\end{proof}

In the next lemma, we show that bispecial factors are preserved under
desubstitution by the Poincar\'e substitution.

\begin{lemma}[\bf P - Antecedent of a bispecial]\label{lem:pbispecial}
Let $u\in\{1,2,3\}^*$ and $w\neq\emptyword$ be a bispecial factor of
$\pi_{jk}(u)$.
Let $v$ be the unique antecedent of $w$ under $\pi_{jk}$. One has 
 either $w=k\pi_{jk}(v)$, or $w=jk\pi_{jk}(v)$.
Furthermore, $v$ is a  bispecial factor of $u$.
\end{lemma}

\begin{proof}
The result is a direct consequence of Lemma~\ref{lem:poincareextensions}.
Since right extensions are preserved by $\pi_{jk}$, we only need to check that
if $w$ has at least two left extensions then so does $v$.

Suppose that $w=k\pi_{jk}(v)$. Remark that $i\notin E^-(w)$.  Thus $j,k\in
E^-(w)$ since $w$ is bispecial.  From Lemma~\ref{lem:poincareextensions},
$k\in E^-(w)$ implies $k\in E^-(v)$. Also, $j\in E^-(w)$ implies that $i\in
E^-(v)$ or $j\in E^-(v)$.  Thus $v$ is bispecial.

Suppose that $w=jk\pi_{jk}(v)$.  Since $j\notin E^-(w)$, then $i,k\in E^-(w)$.
Or course, the existence of $w$ implicitly suppose $j\in E^-(k\pi_{jk}(v))$.
Then, $i,j\in E^-(v)$. We conclude that $v$ is bispecial.
\end{proof}

Now we want to describe more precisely under which conditions a bispecial word
$v$ has a unique bispecial extended image and  provide its extension type as
we were able to do in Lemma~\ref{lem:arbispecial} for Arnoux-Rauzy
substitutions. In general (see Table~\ref{table:poincaredegree2}
and~\ref{table:poincaredegree3}), this depends on its left extensions
$E^-(v)$. However, if the left valence satisfies $d^-(v)=2$, we deduce the
unicity of the bispecial extended image as well as important information on
the extension type of the extended image. Recall that the notion of left
equivalence for  extension types was defined in Section~\ref{sec:bispecial} in
Definition~\ref{def:leftequivalent}.

\begin{lemma}[\bf P - Bispecial extended images in details]\label{lem:detailedpoincare}
Let  $i,j,k$ such that $\{i,j,k\}=\{1,2,3\}$. Let $u \in \{1,2,3\}^*$ and 
let $v$ be a bispecial factor of $u$. We assume that  for all $(a,b) \in E(v)$, there exists a letter
 $e$ such that $eavb$ is also  a factor of $u$.

\begin{enumerate}[\rm (i)]
\item If $d^{-}(v)=2$,  $v$ admits  a unique bispecial extended image
$w\in\{k\pi_{jk}(v),jk\pi_{jk}(v)\}$ under $\pi_{jk}$
and $d^-(w)=2$.  Moreover,
the extension types $E(v)$ and $E(w)$ (in $\pi_{jk}(u)$)  are left equivalent and are related
according to Table~\ref{table:poincaredegree2}.
\item If $d^{-}(v)=3$, then $v$  admits either one, or two bispecial extended images
$w\in\{k\pi_{jk}(v),jk\pi_{jk}(v)\}$ under $\pi_{jk}$. In any case,
$d^{-}(w)=2$ and the two non-empty rows of $E(w)$ are obtained by projection
of rows of $E(v)$. Furthermore, they are  related  according to   Table~\ref{table:poincaredegree3}.
\end{enumerate}
\end{lemma}

\begin{proof}
For each $a\in\{1,2,3\}$, let $R_a\subseteq\{1,2,3\}$  be such that
\[
E(v) = \bigcup_{a\in\{1,2,3\}} \{a\} \times R_a.
\]
The set $R_a$ denotes the right extensions associated with the left extension
$a\in E^-(v)$.\\
(i) 
If $d^{-}(v)=2$, then $E^-(v)$ is equal to either $\{i,j\}$, $\{i,k\}$ or
$\{j,k\}$. 
We proceed case by case.
If $E^-(v)=\{i,j\}$, then $k\pi_{jk}(v)$ is not left special and $jk\pi_{jk}(v)$
is the unique bispecial extended image of $v$.
If $E^-(v)=\{i,k\}$ or $\{j,k\}$,
then $jk\pi_{jk}(v)$ is not left special and $k\pi_{jk}(v)$ is
the unique bispecial extended image of $v$.
This is summarized in Table~\ref{table:poincaredegree2} where the information
follows from Lemma~\ref{lem:poincareextensions}.
\begin{table}[h!]
\[
\begin{array}{c|c|c}
E(v) & E(k\pi_{jk}(v)) & E(jk\pi_{jk}(v)) \\
\hline
(\{i\}\times R_i)\cup(\{j\}\times R_j)  &  \{j\}\times (R_i\cup R_j)
& (\{i\}\times R_i)\cup(\{k\}\times R_j)   \\
(\{i\}\times R_i)\cup(\{k\}\times R_k)  &  (\{j\}\times R_i)\cup(\{k\}\times R_k)   & 
\{i\}\times R_i   \\
(\{j\}\times R_j)\cup(\{k\}\times R_k)  &  (\{j\}\times R_j)\cup(\{k\}\times R_k)   &
\{k\}\times R_j
\end{array}
\]
\caption{If $d^-(v)=2$, then exactly one extended image of $v$ amongst
$k\pi_{jk}(v)$ and $jk\pi_{jk}(v)$ is bispecial. This only depends on the left
extensions as the right extensions are preserved.} \label{table:poincaredegree2}
\end{table}
In each case, the extension type $E(v)$ is left equivalent to the extension
type of the unique bispecial extended image $w$ of $v$. Moreover $d^-(w)=2$.

(ii)
If $d^{-}(v)=3$, i.e., $E^-(v)= \{i,j,k\}$, then
$E^-(k\pi_{jk}(v))=\{j,k\}$ and $E^-(jk\pi_{jk}(v))=\{i,k\}$. Thus, both
extended images can be bispecial but their left valence is at most $2$.
This is summarized
in Table~\ref{table:poincaredegree3}.
\begin{table}[h!]
\[
\begin{array}{c|c|c}
E(v) & E(k\pi_{jk}(v)) & E(jk\pi_{jk}(v)) \\
\hline
(\{i\}\times R_i)\cup(\{j\}\times R_j)\cup(\{k\}\times R_k)  &
(\{j\}\times R_i\cup R_j)\cup(\{k\}\times R_k)   &
(\{i\}\times R_i)\cup(\{k\}\times R_j)
\end{array}
\]
\caption{If $d^-(v)=3$, then one or both extended images of $v$ amongst
$k\pi_{jk}(v)$ and $jk\pi_{jk}(v)$ are bispecial. In each case, their left valence
is $2$.}
\label{table:poincaredegree3}
\end{table}
\end{proof}

Note that Table~\ref{table:poincaredegree2} and~\ref{table:poincaredegree3}
provide much more information than does the statement of
Lemma~\ref{lem:detailedpoincare} and they will be used to prove a more general
result in Lemma~\ref{lem:lifebispecial}. For example, in
Table~\ref{table:poincaredegree3}, if $v$ is a bispecial
factor such that $d^{-}(v)=3$, $R_i=R_j$ and $|R_i|=|R_j|=1$, then
$jk\pi_{jk}(v)$ is a left special factor but not a right special factor,   it is thus 
not bispecial.

\subsection{Life of a bispecial factor under ARP substitutions}\label{subsec:life}

In this section, the life of a bispecial factor is analyzed more precisely
under the application of Arnoux-Rauzy and Poincar\'e substitutions in the
spirit of 
\cite[Section 4.2.2]{MR1440670} where bispecial factors are described under the image of
circular morphisms. To achieve this, we need to understand exactly the left
extensions which will give information about the multiplicity of the bispecial
factors.  

Let $
\S      =  \S_\alpha\cup\S_\pi.$
Let $w$ be a factor of an infinite word Arnoux-Rauzy-Poincar\'e  $\S$-adic  word. Let $w_0=w$ and  $w_{i+1}$ be the unique antecedent of $w_i$
under $\sigma_{i}$ for $i\geq 0$. In particular,  $w_1$ is the antecedent of $w_0$ under $\sigma_0$ and
$w_2$ is the antecedent of $w_1$ under $\sigma_1$.
If $|w_i|>0$, then $|w_{i+1}|<|w_i|$ by Lemma \ref{lem:growth}.  There  thus exists $n$ such
that $w_n=\emptyword$. \begin{definition}[\bf Age, History, Life]
Let $w$ be a factor of an Arnoux-Rauzy-Poincar\'e  $\S$-adic  word.  
Let $w_0=w$ and  $w_{i+1}$ be the unique antecedent of $w_i$
under $\sigma_{i}$ for $i\geq 0$.
The smallest of the integers $n$ for which  $w_n=\emptyword$ is called the \emph{age} of $w$ and is
denoted as  $\age(w)$.
Furthermore, we say that the finite sequence $\sigma_0\sigma_1\cdots\sigma_n$
is the \emph{history} and the sequence $(w_i)_{0\leq i\leq n}$ is the
\emph{life} of the word~$w$.
\end{definition}
\begin{figure}
\begin{center}
\includegraphics{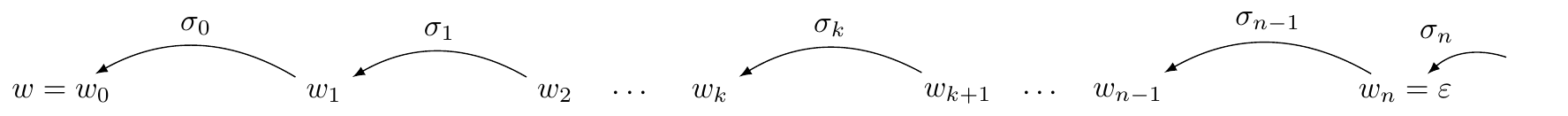}
\end{center}
\caption{Life and history of a factor $w$.}
\label{fig:historic}
\end{figure}
The above definition is illustrated in Figure \ref{fig:historic}.
According to Lemma \ref{lem:arbispecial} and \ref{lem:pbispecial},
all the  words $w_i$ of the history of $w$ are   bispecial factors when $w$ is  bispecial.
We will consider from now on recurrent  Arnoux-Rauzy-Poincar\'e  $\S$-adic  words $\bu$, with 
$\bu^{(m)}$ being  also recurrent,  in order to apply the assumptions of
Lemma  \ref{lem:poincareextensions} and \ref{lem:detailedpoincare}. According to Proposition \ref{prop:recurrence},  note that this assumption applies in particular  to all the words of the 
Arnoux-Rauzy-Poincar\'e  $\S$-adic  system.

\begin{lemma}\label{lem:atmosttwobispecialofsameage}
Let ${\bf u}$ be a recurrent  Arnoux-Rauzy-Poincar\'e  $\S$-adic  word such that $\bu^{(m)}$ is also recurrent for all $m$. 
Let $n\geq 0$ be an integer.
Let $B_n$ be the set of all bispecial factors of age $n$ in~$\bu$.
Then $\Card\,B_n \leq 2$.
\end{lemma}
\begin{proof}
Let $w\in B_n$, $\sigma_0\sigma_1\cdots\sigma_n$ be its
history, and let $(w_i)_{0\leq i\leq n}$ be its life. 

Suppose first that  $\sigma_0\sigma_1\cdots\sigma_n\in\S_\alpha^*\S$, that  is,  the substitutions    of the history of $w$ are all Arnoux-Rauzy substitutions except possibly  $\sigma_n$ which may be a  Poincar\'e substitution.  From
Lemma~\ref{lem:extendedimagesAR}, $w_i$ is the unique
extended image of $w_{i+1}$ for all $0\leq i\leq n-1$. Then $\Card B_n=1$.

Suppose  now that $\sigma_0\sigma_1\cdots\sigma_n\in\S^*\,\pi_{jk}\,\S_\alpha^*\S$.
Let $\ell$ be  the largest index smaller than $n$  of occurrence of  $\pi_{jk}$, that is,  
\[
\sigma_0\sigma_1\cdots\sigma_\ell\in\S^*\,\pi_{jk}
\quad\text{ and }\quad
\sigma_{\ell+1}\sigma_{\ell+2}\cdots\sigma_n\in\S_\alpha^*\S.
\]
Then, from Lemma~\ref{lem:extendedimagesAR}, $w_i$ is the unique extended
image of $w_{i+1}$ for all $\ell+1\leq i\leq n-1$.
Also, from Lemma~\ref{lem:detailedpoincare},
$w_{\ell+1}$ has at most two extended images $w_\ell$ and
$w'_\ell$ in $\{k\pi_{jk}(w_{\ell+1}), jk\pi_{jk}(w_{\ell+1})\}$.
But then, $d^-(w'_\ell)=d^-(w_\ell)=2$ (still by Lemma~\ref{lem:detailedpoincare}). Therefore both
$w'_\ell$ has a unique extended image $w'_{\ell-1}$ and
$w_\ell$ has a unique extended image $w_{\ell-1}$ (by Lemma~\ref{lem:detailedpoincare} (i)).
Recursively, we get
$d^-(w'_i)=d^-(w_i)=2$ for all $0\leq i \leq \ell$, 
$w'_i$ has a unique extended image $w_{i-1}$, and 
$w_i$ has a unique extended image $w_{i-1}$
for all $1\leq i \leq \ell$.
We thus get  $\Card B_n\leq 2$.
\end{proof}

The life $(w_i)_{0\leq i\leq n}$ of bispecial factors ``starts''  (when read backwards with decreasing indices) as the empty word $\emptyword$
at $i=n$. The word $w_i$ for $i<n$ is  then obtained as the concatenation of one or two letters concatenated with 
$\sigma_i(w_{i+1})$. These letters depend on the extension type $E(w_{i+1})$
and recursively on the extension type $E(w_n)$ of $w_n=\emptyword$.  Furthermore, 
$w_n$ is the antecedent of $w_{n-1}$ under $\sigma_{n-1}$ and
the extension type $E(w_n)$ of $w_n=\emptyword$ depends on $\sigma_n$. Thus, it
is important to understand properly what are the possible extension types of
the empty word under the application of Arnoux-Rauzy and Poincar\'e
substitutions.  Below, the extension type $E(\emptyword)$ of the empty word
considered as a bispecial factor in the language of $\sigma(u)$ is denoted by
$E_{\sigma(u)}(\emptyword)$.

\begin{lemma}\label{lem:typeextempty}
Let $\bu\in\A^*\cup\A^\NN$ be  a proper word. Considered as a bispecial factor of the language of the word $\alpha_k(\bu)$,
the empty word $\emptyword$ is ordinary.
Considered as a bispecial factor of the language of the word $\pi_{jk}(\bu)$,
the empty word $\emptyword$ is neutral but not ordinary:
\[
\scriptsize
E_{\alpha_k(\bu)}(\emptyword) =
\begin{array}{c|ccc}
  & i      & j      & k      \\
\hline
i &        &        & \times \\
j &        &        & \times \\
k & \times & \times & \times
\end{array}
\quad\quad
\text{\normalsize and}
\quad\quad
E_{\pi_{jk}(\bu)}(\emptyword) =
\begin{array}{c|ccc}
  & i      & j      & k      \\
\hline
i &        & \times & \\
j &        &        & \times \\
k & \times & \times & \times
\end{array}.
\]
\end{lemma}

\begin{proof}
We need to consider the set of pairs of consecutive letters appearing in the
language $\alpha_k(u)$. These can be consecutive letters
inside $\alpha_k(1)$, $\alpha_k(2)$ or $\alpha_k(3)$, i.e.,
$\{ik, jk\}$. Alternatively, it may be the last letter of a word
$\alpha_k(a)$ with the first letter of a word $\alpha_k(b)$:
$\{ki,kj,kk\}$.

Similarly for the language $\pi_{jk}(u)$, consecutive letters
inside $\pi_{jk}(i)$, $\pi_{jk}(j)$ and $\pi_{jk}(k)$ are $\{ij, jk\}$
and pairs made of the last letter of a word $\pi_{jk}(a)$ with the first
letter of a word $\pi_{jk}(b)$ are $\{ki,kj,kk\}$.
\end{proof}

From now on, we assume that  the Arnoux-Rauzy-Poincar\'e  $\S$-adic  words $\bu ^{(m)}$ are all
proper for all $m$  in order to apply Lemma~\ref{lem:typeextempty} for the  bispecial  factors of all
ages. Note that being recurrent does not imply  the fact of being proper: indeed an  infinite word
can be recurrent on the alphabet $\{1,2\}$ while each letter of the alphabet $\{1,2,3\}$ 
must appear for this word to be  proper.

\begin{lemma}
Let ${\bf u}$ be an  Arnoux-Rauzy-Poincar\'e  $\S$-adic
word such that $\bu^{(m)}$ is  proper and  recurrent for all $m$.
Let $w$ be a bispecial factor of ${\bf u}$.
Then $|E(w)|\leq 5$.
\end{lemma}

\begin{proof}
Let $n=\age(w)$  and $(w_i)_i$ be the  life of $w$.
From Lemmas \ref{lem:arbispecial}, \ref{lem:detailedpoincare} and
\ref{lem:typeextempty}, we have
\[
|E(w_0)|
\leq |E(w_1)|
\leq |E(w_2)|
\leq \cdots
\leq |E(w_n)|
\leq 5.\qedhere
\]
\end{proof}

The following  lemma shows that the  histories  of  bispecial factors  in a   same infinite word ${\bf u}$ are related.
\begin{lemma}\label{lem:common}
Let ${\bf u}$ be an   Arnoux-Rauzy-Poincar\'e  $\S$-adic
word such that $\bu^{(m)}$ is  proper and recurrent for all $m$.  Let $w$ and $z$ be
two bispecial factors of  ${\bf u}$.
\begin{enumerate}[\rm (i)]
\item If $\age(w)<\age(z)$, then the history of $w$ is a prefix
of the one of $z$.
\item If $\age(w)=\age(z)$, then $w$ and $z$ have the same history.
\end{enumerate}
\end{lemma}

\begin{proof}
Statement (i) follows from the definition and  statement 
(ii) follows from (i).
\end{proof}

In the next lemma, we describe exactly what are the bispecial factors
associated with each possible history.  We recall that there are  at most two  bispecial factors of the same age  for a given history 
according to Lemma \ref{lem:atmosttwobispecialofsameage}. It has the same history  as  $w$  according to Lemma~\ref{lem:common}.


\begin{lemma}\label{lem:lifebispecial}
Let $\bu$   be an Arnoux-Rauzy-Poincar\'e  $\S$-adic  word
such that $\bu^{(m)}$ is  proper and  recurrent for all $m$.  Let $w$ be a bispecial factor of $\bu$
 and let  $n=\age(w)$.
Let $w'$ be the other bispecial factor of the same
age as $w$ if it exists.
Then the common  history 
$\sigma_0\sigma_1\cdots\sigma_n$
 of $w$  and $w'$ determines the left
valence,  the multiplicity and the extension type of both  $w$ and $w'$.
More precisely,   the multiplicity and the extension type are described in 
Table~\ref{table:bispecial}, whereas extension types  are provided in
Figures \ref{fig:life_ordinaire_simple}, \ref{fig:life_ordinaire_double},
\ref{fig:life_pas_ordinaire_to_ordinaire} and
\ref{fig:life_pas_ordinaire_strong_weak}.

\begin{table}[h]
\[
\begin{array}{l|ccc|ccc}
\sigma_0\sigma_1\cdots\sigma_n\in
& d^-(w) & m(w) & \text{ordinary}
& d^-(w') & m(w') & \text{ordinary} \\
\hline
\S_\alpha^*\S_\alpha
& 3 & 0 & \yes \\
\S_\alpha^*\S_\pi
& 3 & 0 & \no \\
\S^*\,\pi_{jk}\,\S_\alpha^*\{\alpha_{k}\}
& 2 & 0 & \yes \\
\S^*\,\pi_{jk}\,\S_\alpha^*\{\alpha_{i},\alpha_{j}\}
& 2 & 0 & \yes  & 2 & 0 & \yes \\
\S^*\,\pi_{jk}\,\S_\alpha^*\{\pi_{ji},\pi_{ki},\pi_{ij},\pi_{kj}\}
& 2 & 0 & \yes & 2 & 0 & \yes \\
\S^*\,\pi_{jk}\,\S_\alpha^*\{\pi_{jk},\pi_{ik}\}
& 2 & +1 & \no & 2 & -1 & \no \\
\end{array}
\]
\caption{ Left
valence and   multiplicity  for the  (at most two) bispecial factors of the same age.}
\label{table:bispecial}
\end{table}
\end{lemma}

\begin{remark}\label{rem:pairs}
Recall that  the occurrence of  strong and weak bispecial factors  has an impact  on the  factor complexity. According to Lemma  \ref{lem:lifebispecial},
strong and weak bispecial words  appear in pairs under the application of
Poincar\'e substitutions each time $\pi_{jk}$ is followed by either  $\pi_{jk}$ or
$\pi_{ik}$ for $\{i,j,k\}=\{1,2,3\}$ with possibly some Arnoux-Rauzy
substitutions $\alpha_k$, $k\in\{1,2,3\}$, in between.
\end{remark}

\begin{proof}
In the following proof, elements $(j,k)$ of $E(w)$ are noted $jk$ for short.
We refer below to the  lines of  Table \ref{table:bispecial}.
\medskip

{\bf Line 1.}
 Assume $\sigma_0\sigma_1\cdots\sigma_n\in \S_\alpha^*\S_\alpha$.  According to Lemma  \ref{lem:arbispecial}, the
extension type is preserved by Arnoux-Rauzy substitutions,  which yields
$E(w)=E_{\sigma_n}(\emptyword)$, so that $d^-(w)=3$. Moreover, since
$\sigma_n\in\S_\alpha$, then $E_{\sigma_n}(\emptyword)$ is ordinary and the
multiplicity is $m(w)=0 $ (by Lemma \ref{lem:typeextempty}).  Also, the bispecial extended images are unique
under the application of each substitution $\sigma_i\in\S_\alpha$, by Lemma \ref{lem:extendedimagesAR}.
\medskip

{\bf Line 2.}  Assume 
$\sigma_0\sigma_1\cdots\sigma_n\in \S_\alpha^*\S_\pi$.
The proof is the same as for Line 1 except that the extension type of the
empty word $E_{\sigma_n}(\emptyword)$ is not ordinary because
$\sigma_n\in\S_\pi$  (by Lemma \ref{lem:typeextempty}).
\medskip

{\bf Line 3-6.}  We assume
$\sigma_0\sigma_1\cdots\sigma_n\in \S^*\,\pi_{jk}\,\S_\alpha^*\S$.
Let  $\ell$  be   the largest  index of  occurrence smaller than $n$  of $\pi_{jk}$, that is, 
\[
\sigma_0\sigma_1\cdots\sigma_\ell\in\S^*\,\pi_{jk}
\quad\text{ and }\quad
\sigma_{\ell+1}\sigma_{\ell+2}\cdots\sigma_n\in\S_\alpha^*\S.
\]
The bispecial antecedent of $w$ under the substitution
$\sigma_0\sigma_1\cdots\sigma_{\ell-1}\in\S^*$ is  $w_{\ell}$, and
$w_{\ell+1}$ is the bispecial antecedent of $w_\ell$ under the substitution
$\sigma_\ell=\pi_{jk}$.
Since $\sigma_{\ell+1}\sigma_{\ell+2}\cdots\sigma_n\in\S_\alpha^*\S$, then
$d^-(w_{\ell+1})=3$ and $w_{\ell+1}$ has two extended images
under $\sigma_\ell=\pi_{jk}$.
Moreover, let $w'_\ell$, with $w'_\ell\neq w_\ell$, be the other extended image of $w_{\ell+1}$.
One has $w_{\ell}, w'_\ell \in \{ k\pi_{jk}(w_{\ell+1}),\ jk\pi_{jk}(w_{\ell+1})\}$.
Note that the factor $w'_\ell$ may be bispecial or not (see e.g.  the proof of the case of  Line~3 below).
The end of the proof for lines 3-6 follows the same pattern. In fact, the first part
$\sigma_0\sigma_1\cdots\sigma_{\ell-1}\in\S^*$  is always applied on a
bispecial factor $w_\ell$ or $w'_\ell$ with  left valence satisfying
$d^-(w_\ell)=d^-(w'_\ell)=2$. Therefore, from Lemma~\ref{lem:detailedpoincare}
(i) the extension types of $w=w_0$ and $w_\ell$ are left-equivalent.
Similarly, the extension types of $w'=w'_0$ and $w'_\ell$ are left-equivalent (where  the $w'_i$ are inductively defined  as   extended images).
From Lemma~\ref{lem:equivalence}, the multiplicity, the left valence and the
fact of being strong, weak or ordinary is preserved by left-equivalence.
Below, we suppose $w_{\ell}=k\pi_{jk}(w_{\ell+1})$ and $w'_\ell=jk\pi_{jk}(w_{\ell+1})$.
\medskip

\begin{figure}[h]
\begin{center}
\includegraphics{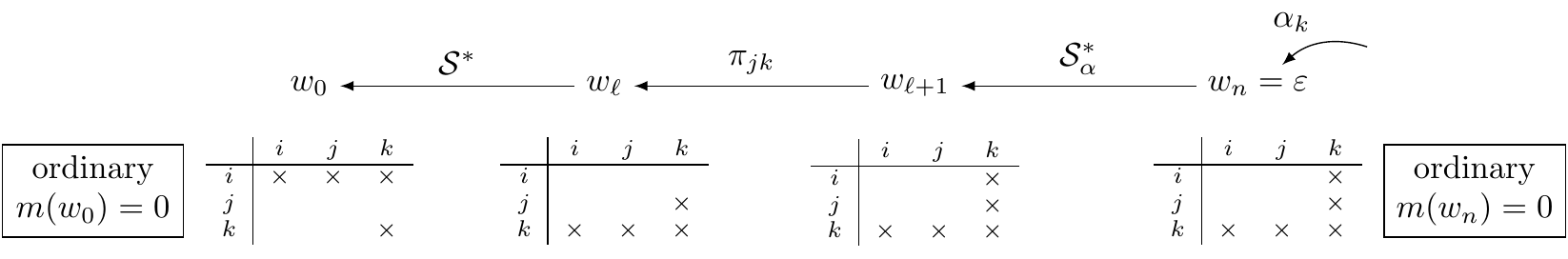}
\end{center}
\caption{Life of a bispecial word if
$\sigma_0\sigma_1\cdots\sigma_n\in
\S^*\,\pi_{jk}\,\S_\alpha^*\{\alpha_{k}\}$.}
\label{fig:life_ordinaire_simple}
\end{figure}
{\bf Line 3.}  We assume 
$\sigma_0\sigma_1\cdots\sigma_n\in
\S^*\,\pi_{jk}\,\S_\alpha^*\{\alpha_{k}\}$
(see Figure~\ref{fig:life_ordinaire_simple}).
If $\sigma_n=\alpha_k$, then
$E(w_{\ell+1})=E_{\sigma_n(u)}(\emptyword)= E_{\alpha_k(u)}(\emptyword)$.
Then from Lemma~\ref{lem:detailedpoincare} (ii) and
Table~\ref{table:poincaredegree3}, we have
$E(w_{\ell+1})=
\{ik,jk,ki,kj,kk\}$,
$E(w_\ell)=
\{jk,ki,kj,kk\}$
and
$E(w'_\ell)=
\{ik,kk\}$.
Then $w_\ell$ is bispecial ordinary and $w'_\ell$ is not bispecial.
\medskip

\begin{figure}[h]
\begin{center}
\includegraphics{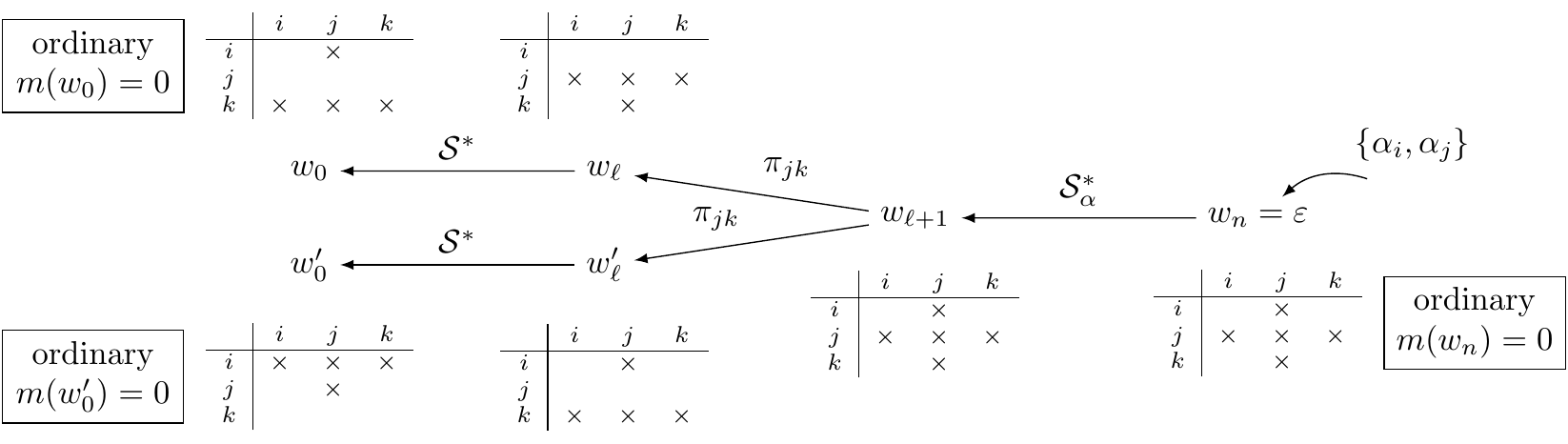}
\end{center}
\caption{Life of a bispecial word if
$\sigma_0\sigma_1\cdots\sigma_n\in
\S^*\,\pi_{jk}\,\S_\alpha^*\{\alpha_{i},\alpha_{j}\}$.
The extension types depicted represent the case $\sigma_n=\alpha_j$.
}
\label{fig:life_ordinaire_double}
\end{figure}
{\bf Line 4.} Assume
$\sigma_0\sigma_1\cdots\sigma_n\in
\S^*\,\pi_{jk}\,\S_\alpha^*\{\alpha_{i},\alpha_{j}\}$
(see Figure~\ref{fig:life_ordinaire_double}).
If $\sigma_n=\alpha_i$, then
$E(w_{\ell+1})=E_{\sigma_n(u)}(\emptyword)= E_{\alpha_i(u)}(\emptyword)$.
Then from Lemma~\ref{lem:detailedpoincare} (ii) and
Table~\ref{table:poincaredegree3}, we have
$E(w_{\ell+1})=
\{ii,ij,ik,ji,ki\}$,
$E(w_\ell)=
\{ji,jj,jk,ki\}$,
and
$E(w'_\ell)=
\{ii,ij,ik,ki\}$.
Then $w_\ell$ and $w'_\ell$ are both bispecial ordinary.

If $\sigma_n=\alpha_j$, then
$E(w_{\ell+1})=E_{\sigma_n(u)}(\emptyword)= E_{\alpha_j(u)}(\emptyword)$.
Then from Lemma~\ref{lem:detailedpoincare} (ii) and
Table~\ref{table:poincaredegree3}, we have
$E(w_{\ell+1})=
\{ij,ji,jj,jk,kj\}$,
$E(w_\ell)=
\{ji,jj,jk,kj\}$
and
$E(w'_\ell)=
\{ij,ki,kj,kk\}$.
Then $w_\ell$ and $w'_\ell$ are both bispecial ordinary.
\medskip

\begin{figure}[h]
\begin{center}
\includegraphics{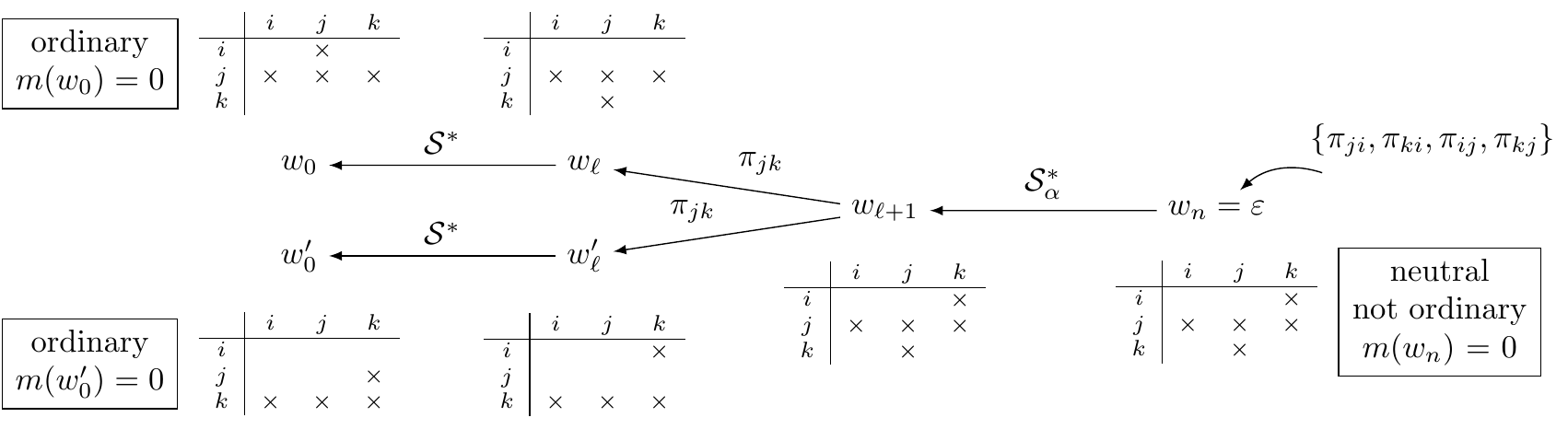}
\end{center}
\caption{Life of a bispecial word if
$\sigma_0\sigma_1\cdots\sigma_n\in
\S^*\,\pi_{jk}\,\S_\alpha^*\{\pi_{ji},\pi_{ki},\pi_{ij},\pi_{kj}\}$.
The extension types depicted represent the case $\sigma_n=\pi_{kj}$.}
\label{fig:life_pas_ordinaire_to_ordinaire}
\end{figure}
{\bf Line 5.}  Assume
$\sigma_0\sigma_1\cdots\sigma_n\in
\S^*\,\pi_{jk}\,\S_\alpha^*\{\pi_{ji},\pi_{ki},\pi_{ij},\pi_{kj}\}$
(see Figure~\ref{fig:life_pas_ordinaire_to_ordinaire}).
If $\sigma_n=\pi_{ji}$, then
$E(w_{\ell+1})=E_{\sigma_n(u)}(\emptyword)= E_{\pi_{ji}(u)}(\emptyword)$.
Then from Lemma~\ref{lem:detailedpoincare} (ii) and
Table~\ref{table:poincaredegree3}, we have
$E(w_{\ell+1})=
\{ii,ij,ik,ji,kj\}$,
$E(w_\ell)=
\{ji,jj,jk,kj\}$
and
$E(w'_\ell)=
\{ii,ij,ik,ki\}$.
Then $w_\ell$ and $w'_\ell$ are both bispecial ordinary.

If $\sigma_n=\pi_{ki}$, then
$E(w_{\ell+1})=E_{\sigma_n(u)}(\emptyword)= E_{\pi_{ki}(u)}(\emptyword)$.
Then from Lemma~\ref{lem:detailedpoincare} (ii) and
Table~\ref{table:poincaredegree3}, we have
$E(w_{\ell+1})=
\{ii,ij,ik,jk,ki\}$,
$E(w_\ell)=
\{ji,jj,jk,ki\}$
and
$E(w'_\ell)=
\{ii,ij,ik,kk\}$.
Then $w_\ell$ and $w'_\ell$ are both bispecial ordinary.

If $\sigma_n=\pi_{ij}$, then
$E(w_{\ell+1})=E_{\sigma_n(u)}(\emptyword)= E_{\pi_{ij}(u)}(\emptyword)$.
Then from Lemma~\ref{lem:detailedpoincare} (ii) and
Table~\ref{table:poincaredegree3}, we have
$E(w_{\ell+1})=
\{ij,ji,jj,jk,ki\}$,
$E(w_\ell)=
\{jk,jj,jk,ki\}$
and
$E(w'_\ell)=
\{ij,ki,kj,kk\}$.
Then $w_\ell$ and $w'_\ell$ are both bispecial ordinary.

If $\sigma_n=\pi_{kj}$, then
$E(w_{\ell+1})=E_{\sigma_n(u)}(\emptyword)= E_{\pi_{kj}(u)}(\emptyword)$.
Then from Lemma~\ref{lem:detailedpoincare} (ii) and
Table~\ref{table:poincaredegree3}, we have
$E(w_{\ell+1})=
\{ik,ji,jj,jk,kj\}$,
$E(w_\ell)=
\{ji,jj,jk,kj\}$
and
$E(w'_\ell)=
\{ik,ki,kj,kk\}$.
Then $w_\ell$ and $w'_\ell$ are both bispecial ordinary.
\medskip

\begin{figure}[h]
\begin{center}
\includegraphics{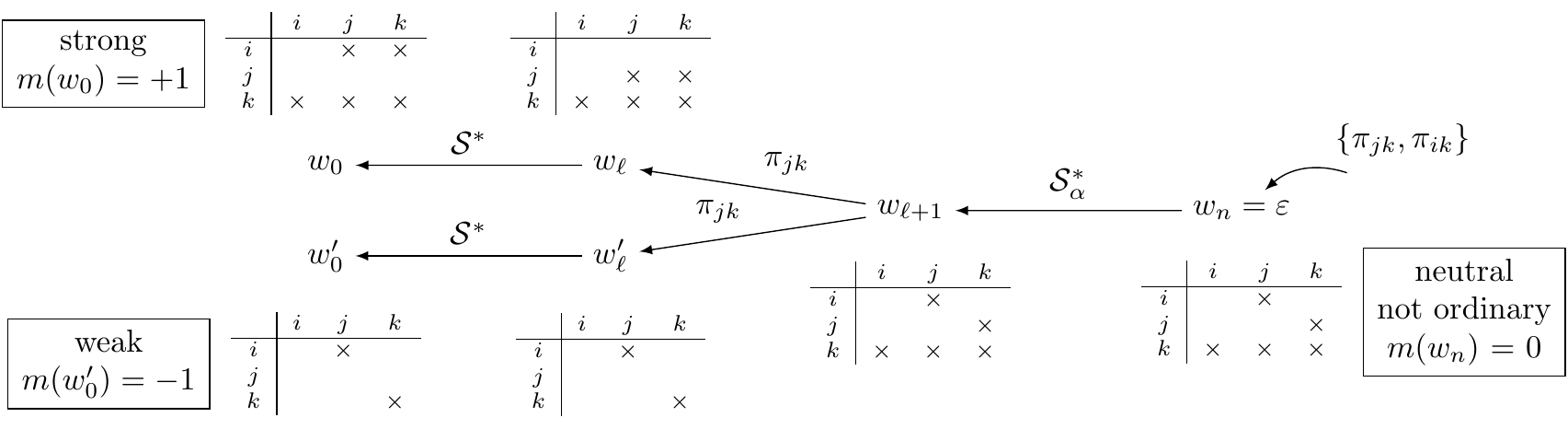}
\end{center}
\caption{Life of a bispecial word if
$\sigma_0\sigma_1\cdots\sigma_n\in
\S^*\,\pi_{jk}\,\S_\alpha^*\{\pi_{jk},\pi_{ik}\}$.
The extension types shown represent the case $\sigma_n=\pi_{jk}$.
}
\label{fig:life_pas_ordinaire_strong_weak}
\end{figure}
{\bf Line 6.} Assume
$\sigma_0\sigma_1\cdots\sigma_n\in
\S^*\,\pi_{jk}\,\S_\alpha^*\{\pi_{jk},\pi_{ik}\}$
(see Figure~\ref{fig:life_pas_ordinaire_strong_weak}).
If $\sigma_n=\pi_{jk}$, then
$E(w_{\ell+1})=E_{\sigma_n(u)}(\emptyword)= E_{\pi_{jk}(u)}(\emptyword)$.
Then from Lemma~\ref{lem:detailedpoincare} (ii) and
Table~\ref{table:poincaredegree3}, we have
$E(w_{\ell+1})=
\{ij,jk,ki,kj,kk\}$,
$E(w_\ell)=
\{jj,jk,ki,kj,kk\}$
and
$E(w'_\ell)=
\{ij,kk\}$.
Then $w_\ell$ is bispecial strong and $w'_\ell$ is bispecial weak.

If $\sigma_n=\pi_{ik}$, then
$E(w_{\ell+1})=E_{\sigma_n}(u)(\emptyword)= E_{\pi_{ik}(u)}(\emptyword)$.
Then from Lemma~\ref{lem:detailedpoincare} (ii) and
Table~\ref{table:poincaredegree3}, we have
$E(w_{\ell+1})=
\{ik,ji,ki,kj,kk\}$,
$E(w_\ell)=
\{ji,jk,ki,kj,kk\}$
and
$E(w'_\ell)=
\{ik,ki\}$.
Then $w_\ell$ is bispecial strong and $w'_\ell$ is bispecial weak.
\end{proof}

\subsection{Quadratic complexity is achievable}\label{sec:quadratic}

According to Remark \ref{rem:pairs}, each time $\pi_{jk}$ and $\pi_{ik}$ are found one next to the other in a certain
$\S$-adic sequence, a new pair of strong and weak bispecial factors is created
(see Lemma~\ref{lem:lifebispecial})
and the length of a newly created  weak bispecial factor can be larger than the length of an
older strong bispecial factor. Therefore, the complexity can increase by more
than $3$, i.e., $p(n+1)-p(n)>3$ for some values of $n$. Let us illustrate  it on the following example. Let
\begin{align*}
    u &= \pi_{23}\pi_{23}
    \pi_{13}\pi_{23}\pi_{23}\alpha_{1}\alpha_{3}\alpha_{2}(1)\\
    &= 1232333233123233332331232333333123233323\cdots
\end{align*}
The  bispecial factors of $u$ of age $\leq5$ and their life are shown in 
Figure~\ref{fig:larger3n}. We see that some weak bispecial factors are longer
than older strong bispecial factors. Because of this fact
and of Equation~\eqref{eq:bisp_mult}, the non-zero values
of the sequence $(b(n))_n$ do not alternate in the set $\{+1,-1\}$.
Therefore, there are values of $n$ for which $s(n)>3$.
The complete computation of $b(n)$, $s(n)$ and $p(n)$ for $n\leq10$ is given
in Table~\ref{tab:larger3n}.
The complexity $p(n)$ of the finite word $u$ satisfies 
$p(n+1)-p(n)=4$ for some values of $n$ and 
$p(n)>3n+1$ for $n$ such that $7\leq n\leq 17$
(recall Equations \eqref{eq:defsn}, \eqref{eq:defbn},
\eqref{eq:pn}, \eqref{eq:sn}, \eqref{eq:bisp_mult}
and in particular that
$s(n)=2+\sum_{\ell=0}^{n-1}b(\ell)$ when the size of alphabet is $3$).
\begin{figure}[h]
\begin{center}
\includegraphics{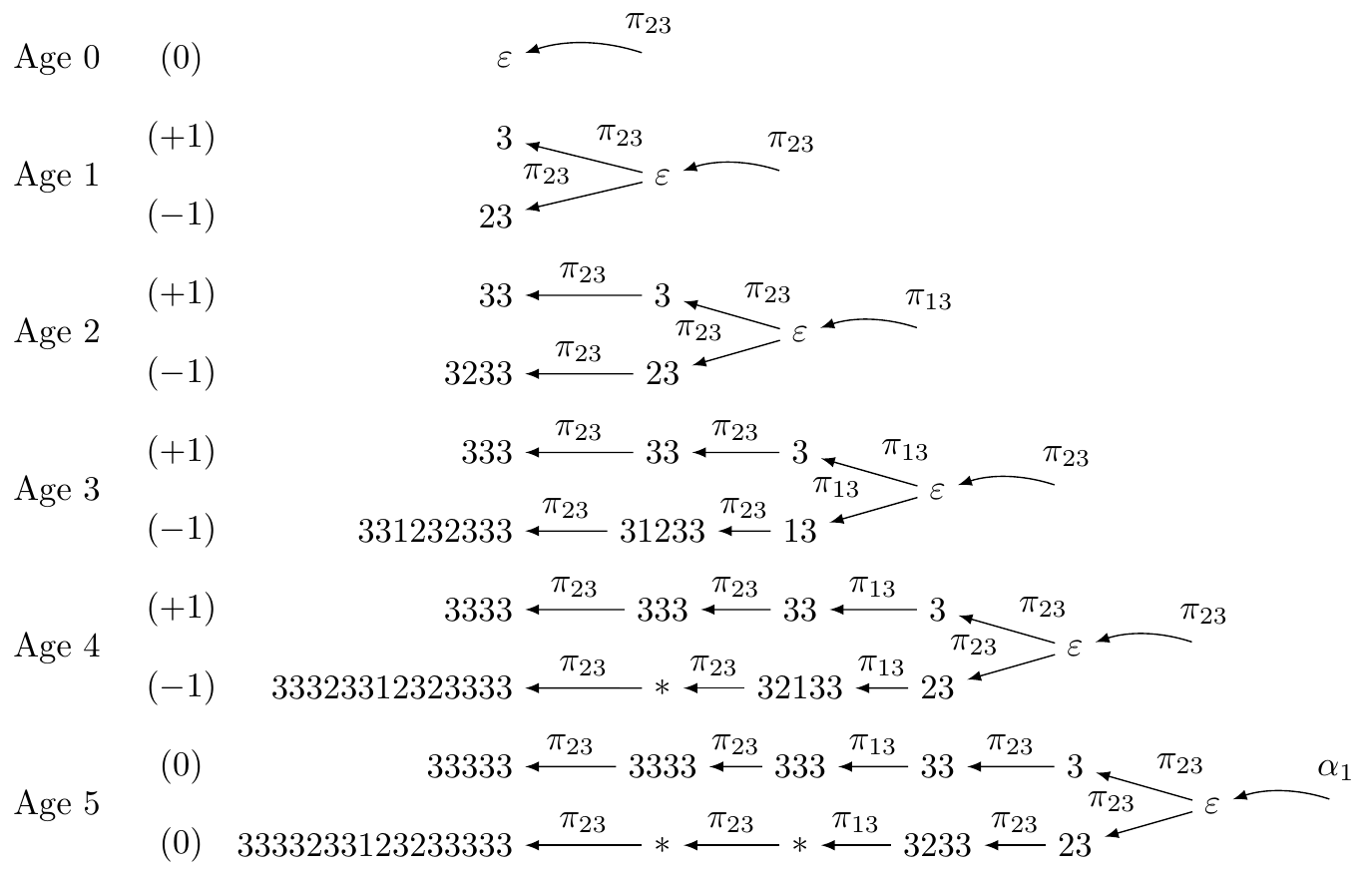}
\end{center}
\caption{
The young bispecial factors of 
$u = \pi_{23}\pi_{23}\pi_{13}\pi_{23}\pi_{23}\alpha_{1}\alpha_{3}\alpha_{2}(1)$
of age $\leq 5$ and their life. The factors $3$, $33$, $333$, $3333$ are strong while $33333$
is neutral. For  the age  $k$  equal to $2$ and $3$, the weak factor of age $k$ is longer than the strong
factor of age $k+1$. The quantity under parenthesis indicates the value of $m$.
}
\label{fig:larger3n}
\end{figure}
\begin{table}[h!]
\[
\begin{array}{c|c|c|c|c|c|c}
    n & \displaystyle \sum_{\substack{w\in\L_n(u)\\ w\,\text{is strong}}} m(w)
      & \displaystyle \sum_{\substack{w\in\L_n(u)\\ w\,\text{is weak}}}   m(w) 
      & b(n) & s(n) & p(n) & 3n+1\\
\hline
0 &  0 & 0  & 0 & 2 & 1  & 1   \\
1 & +1 & 0  & 1 & 2 & 3  & 4   \\
2 & +1 & -1 & 0 & 3 & 5  & 7   \\
3 & +1 & 0  & 1 & 3 & 8  & 10   \\
4 & +1 & -1 & 0 & 4 & 11 & 13  \\
5 &  0 & 0  & 0 & 4 & 15 & 16  \\
6 &  0 & 0  & 0 & 4 & 19 & 19  \\
7 &  0 & 0  & 0 & 4 & 23 & 22  \\
8 &  0 & 0  & 0 & 4 & 27 & 25  \\
9 &  0 & -1 & -1 & 4 & 31 & 28  \\
10 & 0 & 0  & 0 & 3 & 35 & 31  \\
\end{array}
\]
\caption{The lengths of strong bispecial factors of $u$ are $1$, $2$, $3$, $4$
    while the  lengths of weak bispecial factors of $u$ are $2$, $4$, $9$, $14$.
    Since there are two more strong bispecial factors of length $\leq n$ than
    the number of weak bispecial factors of length $\leq n$ for all $n$ such
    that $3\leq n\leq8$, then $s(n)=4$ for each $n$ with $4\leq n\leq9$.  For
    example, $s(4)=p(5)-p(4)=s(0)+\sum_{\ell=0}^{3}b(\ell)=4$. Moreover,
    $p(7)=23>22=3\cdot 7+1$.}
\label{tab:larger3n}
\end{table}

In fact the complexity can get higher.  It follows from
Theorem 4.7.66 of \cite[p. 214]{MR2759107} that the fixed point of
$\pi_{23}\pi_{13}$ starting with letter $1$ has a quadratic factor complexity
because it has infinitely many distinct factors, namely the factors $3^n$,
that are  bounded (in fact fixed) under $\pi_{23}\pi_{13}$.



\subsection{Partial and strict partial order on $\RR^3$}\label{sec:partialorder}

In this section, we consider two distinct partial orders on $\RR^3$ and
consider how these partial orders are preserved by the application of
Arnoux-Rauzy and Poincar\'e substitutions. The results allow the understanding of
the growth of bispecial factors and are used in the proof of  Theorem \ref{thm:leq3n}
in the next section.

Let $\vect{u}=(u_1,u_2,u_3), \vect{v}=(v_1,v_2,v_3)\in\NN^3$ be two  abelianized vectors (for two words $u,v$).
We define $<$ as the strict partial order (irreflexive, transitive and
thus asymmetric) defined coordinate per coordinate on $\NN^3$ by:
\[
\vect{u} < \vect{v}
\iff
u_1 < v_1\quad\text{and}\quad
u_2 < v_2\quad\text{and}\quad
u_3 < v_3.
\]
Also, we define $\leq$ as the partial order (reflexive, transitive and antisymmetric)
defined coordinate per coordinate on $\NN^3$:
\[
\vect{u} \leq \vect{v}
\iff
u_1 \leq v_1\quad\text{and}\quad
u_2 \leq v_2\quad\text{and}\quad
u_3 \leq v_3.
\]
Moreover, we say that the inequality $\vect{u} \leq \vect{v}$ is \emph{strict on
the index $i$} if $u_i<v_i$.
Note that $\leq$ is not the reflexive closure of $<$ since it includes more
relations.

The next lemma shows that the relation $<$ is preserved by Arnoux-Rauzy and
Poincar\'e substitutions and that some stronger conditions are satisfied.
These stronger conditions are used to show at
Lemma~\ref{lem:partialorderchignonpreserved} that the relation $<$ is also
preserved for extended images of factors.
In the next lemma and the next sections, we fix $\be_1=(1,0,0)$, $\be_2=(0,1,0)$ and $\be_3=(0,0,1)$.

\begin{lemma}\label{lem:partialorderpreserved}
Let $v,v'\in\A^*$ be such that $\vect{v}<\vect{v'}$.
For all $\{i,j,k\}=\{1,2,3\}$,
\begin{enumerate}[\rm (i)]
\item $\vect{\alpha_{k}(v)}+2\be_k< \vect{\alpha_{k}(v')}$,
\item $\vect{\pi_{jk}(v)}+\be_j+2\be_k< \vect{\pi_{jk}(v')}$.
\end{enumerate}
In particular, if $\vect{v}<\vect{v'}$ then
$\vect{\alpha_{k}(v)}< \vect{\alpha_{k}(v')}$ and
$\vect{\pi_{jk}(v)}< \vect{\pi_{jk}(v')}$.
\end{lemma}

The proof is in the appendix.

The next lemma shows that the relation $<$ is preserved by Arnoux-Rauzy and
Poincar\'e substitutions from a pair of factors to their extended images.

\begin{lemma}\label{lem:partialorderchignonpreserved}
Let $\sigma\in\S$.  Let $v,v',w,w'\in\A^*$ and suppose
$w$ (resp. $w'$) is an extended image of $v$ (resp. $v'$) under
$\sigma$. If $\vect{v}<\vect{v'}$, then $\vect{w}<\vect{w'}$.
\end{lemma}

The proof is in the appendix.
\begin{remark}
The previous lemma is false for the order $\leq$.
Indeed $\pi_{jk}$ does not preserve the relation $\leq$ for extended images.
For example, if $v=\emptyword$ and $v'=3$, then
$\vect{v}=(0,0,0)\leq(0,0,1)=\vect{v'}$
but
\[
\vect{13\pi_{13}(v)}=\vect{13}=(1,0,1)\not\leq (0,0,2)=\vect{33}=\vect{3\pi_{13}(v')},
\]
and this may even lead after some more substitutions to an inversion of the
order:
\[
\vect{3\pi_{23}(13)}=\vect{31233}=(1,1,3)\geq (0,0,3)=\vect{333}=\vect{3\pi_{23}(33)}.
\]
This example can be seen between age $3$ and $4$ in Figure~\ref{fig:larger3n}.
\end{remark}

\section{Proof of  Theorem  \ref{thm:leq3n}}\label{sec:proofthm}

We now consider $\S$-adic words $\bu$ generated by the Arnoux-Rauzy-Poincar\'e
algorithm applied to a totally irrational vector $\bx\in\Delta$.  By
Proposition~\ref{prop:primitive},   $\bu^{(m)}$ is   proper and  uniformly
recurrent for all $m$ so
the hypothesis  introduced  in the previous section is satisfied.  Such sequences
are in the Arnoux-Rauzy-Poincar\'e $\S$-adic system (Type $3$), that is, we
take into account the restrictions on the directive sequences provided by
Proposition \ref{prop:1}. 
The examples in Section~\ref{sec:quadratic}  show that
Arnoux-Rauzy-Poincaré $\S$-adic sequences can lead in general to quadratic
factor complexity.
Nevertheless 
we show that the factor complexity $p(n)$ of $\S$-adic words $\bu$ generated by the
Arnoux-Rauzy-Poincar\'e algorithm applied to a totally irrational vector 
satisfy $p(n+1)-p(n)\in\{2,3\}$. 
Thus, their factor complexity is bounded below and
above,  that is,  $2n+1\leq p(n)\leq 3n+1$ for all $n$.  In fact, we even  prove  that $p(n+1)-p(n)$
is equal to $2$ more often than it is equal to $3$ which implies that
$p(n)\leq\frac{5}{2}n+1$.
More precisely,  we   will  show that strong and weak bispecial words
alternate when the length increases in Section \ref{subsec:2,3}. We then   consider  more closely the lengths of consecutive values of  $2$ and $3$ in 
the sequence $(p(n+1)-p(n))_n$ in Section \ref{subsec:52}. By making  use  of
Lemma~\ref{lem:pn23iffsumbn01} together with Lemma \ref{lem:alternate}  (see Figure~\ref{fig:pairebispecial}),  we will be  able to prove Theorem 
 \ref{thm:leq3n} in Section  \ref{subsec:proof}.


\subsection{Alternance of  strong and weak  bispecial factors} \label{subsec:2,3}

We first gather the lemmas required in the  proof (see Section
\ref{subsec:proof})  of the fact that the  $\S$-adic words  $\bu$ (with  $\bu ^{(m)}$ recurrent for all $m$) such  that
$\sigma_k\sigma_{k+1}\cdots\sigma_{\ell}\in \L(\G)$ (for all $k, \ell$)   provide
 words that satisfy $p(n+1)-p(n)\in\{2,3\}$.

Restricted to the language of the automaton $\G$, illustrated in
Figure~\ref{figure:markovchain}, the history of a strong or weak bispecial
factor necessarily contains Arnoux-Rauzy substitutions.

\begin{lemma}\label{lem:weak-stron-language-restricted}
Let $\bu=\lim_{n\to\infty}\sigma_0\sigma_1\cdots\sigma_n(a_n)$  be an
$\S$-adic word generated by the Arnoux-Rauzy-Poincar\'e
algorithm applied to a totally irrational vector $\bx\in\Delta$.
Let $w$ be a bispecial factor of $\bu$ and let  $n=\age(w)$.

If $w$ is weak or strong
and the history of $w$ is in the regular language
$\sigma_0\sigma_1\cdots\sigma_n \in \L(\G)$, then
\[
\sigma_0\sigma_1\cdots\sigma_n\in
\S^*\, \pi_{jk} \{\alpha_j\}^*\, \alpha_i\, \S_\alpha^*\,\, \{\pi_{ik},\pi_{jk}\}
\]
for some $\{i,j,k\}=\{1,2,3\}$.
\end{lemma}

\begin{proof}
From Lemma~\ref{lem:lifebispecial}, we have
\[
\sigma_0\sigma_1\cdots\sigma_n\in
\S^*\,\pi_{jk}\,\S_\alpha^*\,\{\pi_{ik},\pi_{jk}\}
\]
for some $\{i,j,k\}=\{1,2,3\}$.
Let $p\in\S^*\,\pi_{jk}$ and $q\in\S_\alpha^*\,\{\pi_{ik},\pi_{jk}\}$ such
that $pq=\sigma_0\sigma_1\sigma_2\cdots\sigma_n$.
The word $p$ starts at the initial state $\Delta$ and ends in the state $H_{jk}$, the
word $q$ starts from the state $H_{jk}$ and ends in state $H_{jk}$ or $H_{ik}$ (see
Figure~\ref{fig:sousautomate}).
\begin{figure}[h]
\begin{center}
\includegraphics{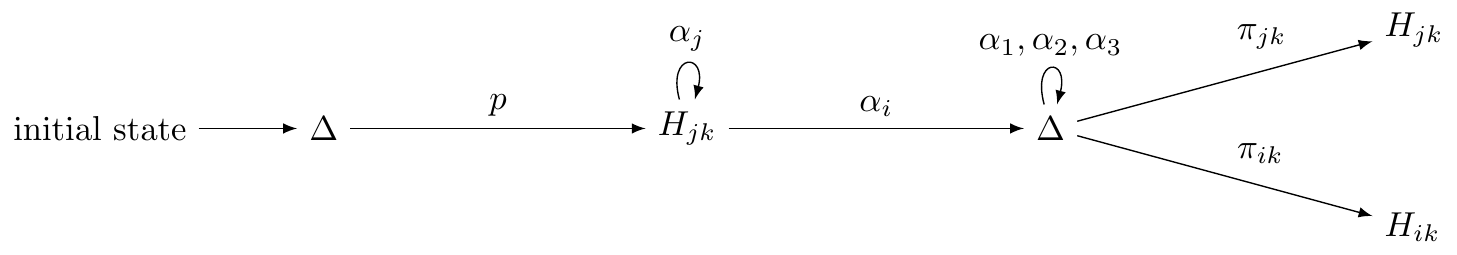}
\end{center}
\caption{The subautomaton of $\G$ describing a path
$\sigma_0\sigma_1\cdots\sigma_n=pq$ such that $p\in
\S^*\,\pi_{jk}$ and $q\in\S_\alpha^*\,\{\pi_{ik},\pi_{jk}\}$.}
\label{fig:sousautomate}
\end{figure}
In the automaton $\G$, the possible transitions  issued from state $H_{jk}$ are
$\pi_{ij}$, $\pi_{ji}$, $\pi_{ki}$, $\alpha_j$ and $\alpha_i$ where only
$\alpha_j$ (looping on state $H_{jk}$) and $\alpha_i$ (going to state
$\Delta$) are allowed by $q\in\S_\alpha^*\,\{\pi_{ik},\pi_{jk}\}$. Once in
state $\Delta$, $q$ allows loops for each symbol in $\S_\alpha$, and finally  the transitions
$\pi_{jk}$ or $\pi_{ik}$ (see Figure~\ref{fig:sousautomate}).  It follows from
this that
\[
q\in\{\alpha_j\}^*\, \alpha_i\, \S_\alpha^*\,\, \{\pi_{ik},\pi_{jk}\}
\]
which was to be proved.
\end{proof}

\begin{lemma}\label{lem:length-lower-bound}
Let $w$ be a bispecial factor   of  an Arnoux-Rauzy-Poincar\'e  $\S$-adic  word.
If for some $\{i,j,k\}=\{1,2,3\}$,
\[
\history(w)\in \pi_{jk}\, \S^* \alpha_i\, \S^* \S,
\]
then $\vect{w}\geq(1,1,1)$.
\end{lemma}

\begin{figure}[h]
\begin{center}
\includegraphics{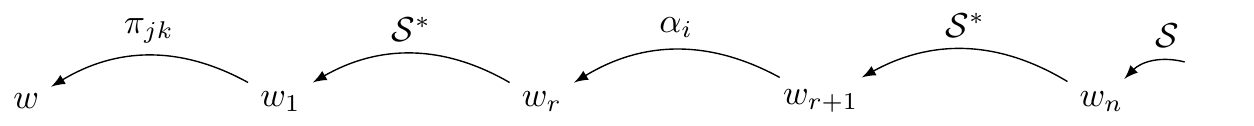}
\end{center}
\caption{We suppose here that
$\history(w)\in \pi_{jk}\, \S^* \alpha_i\, \S^* \S$.}
\label{fig:vectorgraph}
\end{figure}

\begin{proof}
Let $w_1$ be the ancestor of $w$ under $\pi_{jk}$.
Let $r$ and $n$ be integers such that $1\leq r < n=\age(w)$ and $w_{r+1}$ is
the ancestor of $w_r$ under substitution $\alpha_i$ as depicted in
Figure~\ref{fig:vectorgraph}. We have that $\vect{w_n}=(0,0,0)$. Also,
$\vect{w_{r+1}}\geq (0,0,0)$ but $w_r=i\alpha_i(w_{r+1})$ contains at least one occurence of the
letter $i$. Then, $w_1$ also contains at least one occurrence of the
letter $i$. Therefore $\vect{w}\geq(1,1,1)$, because $\pi_{jk}$ maps
$i$ to $ijk$.
\end{proof}

In order to prove that $p(n+1)-p(n)\in\{2,3\}$ for 
 Arnoux-Rauzy-Poincar\'e  $\S$-adic  words $\bu$  such that $\bu^{(m)}$
is  proper and recurrent for all $m$, it is sufficient that strong and weak bispecial words
alternate when the length increases because of Lemma~\ref{lem:pn23iffsumbn01}.  More precisely, if
$z_1$ and $z_3$ are two strong (with multiplicity $+1$) bispecial factors of a word
$\bu$ such that $|z_1|<|z_3|$, then there exists a weak (with multiplicity $-1$)
bispecial factor $z_2$ such that $|z_1|<|z_2|\leq|z_3|$. Note that the notion of alternance was also 
used to prove Theorem 4.11.2 in \cite[p.  238]{MR2759107}.

\begin{lemma}\label{lem:weak_larger_than_strong}
Let $\bu$
be an Arnoux-Rauzy-Poincar\'e  $\S$-adic  word such that $\bu^{(m)}$
is  proper and recurrent for all $m$. 
Let $z^+$ and $z^-$ be two bispecial factors of $\bu$ of the same age.
 Suppose  that $z^-$ is weak and $z^+$ is strong. Then
$|z^+|<|z^-|$. 
\end{lemma}

\begin{proof}
In this proof, we denote by $\vect{z^+}\leq_j \vect{z^-}$ when $\vect{z^+}\leq
\vect{z^-}$ is strict on the coordinate $j\in\{1,2,3\}$.

We prove by induction on the age of bispecial factors that
$\vect{z^+}\leq_j\vect{z^-}$ is strict on
at least one coordinate $j\in\{1,2,3\}$ with $j\in E^-(z^+)$.

Let us prove  the base step of the induction. Suppose that  $z^+$ and $z^-$
have a common neutral bispecial antecedent $v$ thus under the substitution
$\pi_{jk}$ for some $\{i,j,k\}=\{1,2,3\}$. Then, $z^+=k\pi_{jk}(v)$ and
$z^-=jk\pi_{jk}(v)$ so that $\vect{z^+}\leq_j\vect{z^-}$ is strict on the
coordinate $j$. Moreover $E^-(z^+)=\{j,k\}$ and $E^-(z^-)=\{i,k\}$ and hence
$j\in E^-(z^+)$.

Suppose now that  $z^+_h$ and $z^-_h$ are two respectively strong and weak
bispecial factors of a word $u$ of the same age such that
$\vect{z^+_h}\leq_k\vect{z^-_h}$ is strict on at least one coordinate $k\in
E^-(z_h^+)$.  Let $z^+_{h-1}$ and $z^-_{h-1}$ be respectively  the unique
bispecial extended  images of $z^+_h$ and $z^-_h$ under the application of
$\sigma_{h-1}$.  We want to show the following implication for proving the
induction:
\[
\vect{z^+_h}\leq_k\vect{z^-_h} \text{ and } k\in E^-(z_h^+)
\implies
\text{there exists $j$ such that }
\vect{z^+_{h-1}}\leq_j\vect{z^-_{h-1}} \text{ and } j\in E^-(z_{h-1}^+).
\]

Since the letters prepended to the left of bispecial extended images depend on
the left extensions by Table~\ref{table:poincaredegree2},
if $E^-(z^-_h)=E^-(z^+_h)$, it is clear that
$E^-(z^-_{h-1})=E^-(z^+_{h-1})$ and
$\vect{z^+_{h-1}}\leq_j\vect{z^-_{h-1}}$ is strict
for some letter $j\in E^-(z^+_{h-1})$.
Suppose now that $E^-(z^-_h)\neq E^-(z^+_h)$
and suppose without lost of generality that
$E^-(z^+_h)=\{2,3\}$ and $E^-(z^-_h)=\{1,3\}$.
The possible cases depending on $\sigma_{h-1}$ are described in the following
table:
\[
\begin{array}{c|cc|cc|cc}
&&&&& \text{if }k=2 & \text{if }k=3\\
\sigma_{h-1} & z^+_{h-1} & E^-(z^+_{h-1})  & z^-_{h-1} & E^-(z^-_{h-1}) &
    \{j\mid \vect{z^+_{h-1}}\leq_j\vect{z^-_{h-1}}\} &
    \{j\mid \vect{z^+_{h-1}}\leq_j\vect{z^-_{h-1}}\}\\
\hline
\alpha_1 &  1\alpha_1(z^+_h) & \{2,3\} &  1\alpha_1(z^-_h) & \{1,3\} & \{1,2  \}& \{1,3\}\\  
\alpha_2 &  2\alpha_2(z^+_h) & \{2,3\} &  2\alpha_2(z^-_h) & \{1,3\} & \{2    \}& \{2,3\}\\
\alpha_3 &  3\alpha_3(z^+_h) & \{2,3\} &  3\alpha_3(z^-_h) & \{1,3\} & \{2,3  \}& \{3\}\\
\pi_{12} &  2\pi_{12}(z^+_h) & \{1,2\} & 12\pi_{12}(z^-_h) & \{2,3\} & \{1,2  \}& \{1,2,3\}\\
\pi_{32} &  2\pi_{32}(z^+_h) & \{2,3\} & 32\pi_{32}(z^-_h) & \{1,2\} & \{2,3  \}& \{2,3\}  \\
\pi_{13} &  3\pi_{13}(z^+_h) & \{1,3\} &  3\pi_{13}(z^-_h) & \{1,3\} & \{1,2,3\}& \{3\}    \\
\pi_{23} &  3\pi_{23}(z^+_h) & \{2,3\} &  3\pi_{23}(z^-_h) & \{2,3\} & \{2,3  \}& \{3\}    \\
\pi_{21} & 21\pi_{21}(z^+_h) & \{1,3\} &  1\pi_{21}(z^-_h) & \{1,2\} & \{1    \}& \{1,3\}  \\
\pi_{31} & 31\pi_{31}(z^+_h) & \{1,2\} &  1\pi_{31}(z^-_h) & \{1,3\} & \{1,2  \}& \{1\}
\end{array}
\]
We check that for all nine possible values of
$\sigma_{h-1}\in\S$, we always have that
$\vect{z^+_{h-1}}\leq_j\vect{z^-_{h-1}}$ is strict for some $j\in E^-(z^+_{h-1})$.
Since we proved $\vect{z^+}\leq_j\vect{z^-}$ is strict on at least one coordinate $j$,
then we conclude that $|z^+|<|z^-|$. 
\end{proof}

\begin{lemma}\label{lem:alternate}
Let $\bu=\lim_{n\to\infty}\sigma_0\sigma_1\cdots\sigma_n(a_n)$  be an
$\S$-adic word generated by the Arnoux-Rauzy-Poincar\'e
algorithm applied to a totally irrational vector $\bx\in\Delta$.
Let $z^-$ and $w^+$ be two bispecial factors of  ${\bf u}$ such that $z^-$ is weak,
 $w^+$ is strong, and the history of both $z^-$ and $w^+$ are in the regular language
$\sigma_0\sigma_1\cdots\sigma_n \in \L(\G)$.
If $\age(z^-)<\age(w^+)$, then $|z^-|<|w^+|$.
\end{lemma}

\begin{figure}[h!]
\begin{center}
\includegraphics[width=1.0\linewidth]{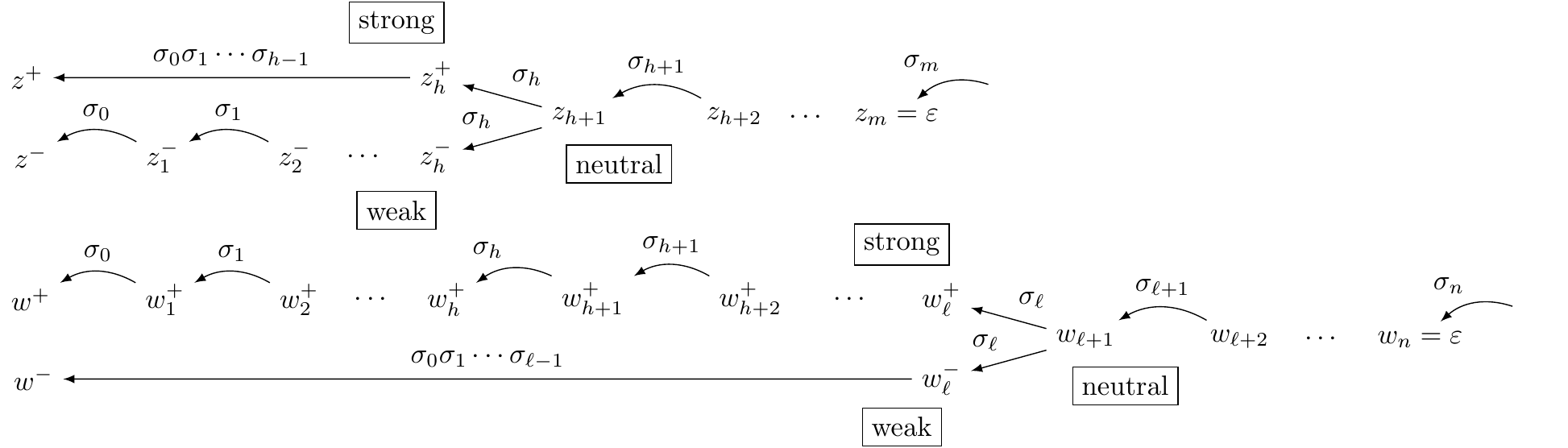}
\end{center}
\caption{Lifes of two pairs of strong and weak bispecial factors: $z^+$,
$z^-$ and $w^+$, $w^-$.}
\label{fig:pairebispecial}
\end{figure}

\begin{proof}
In this proof, we denote a bispecial factor $w$ as $w^+$ if it is strong,
$w^-$ if it is weak, and with no sign if it is neutral.
Let $m=\age(z^-)$, $z^-_0=z^-$ and  $z_{i+1}$ be the unique antecedent of
$z_i$ under $\sigma_{i}$ for $0\leq i\leq m-1$.
Let $n=\age(w^+)$, $w^+_0=w^+$ and  $w_{i+1}$ be the unique antecedent of
$w_i$ under $\sigma_{i}$ for $0\leq i\leq n-1$.
From Lemma~\ref{lem:weak-stron-language-restricted}, we have
\[
\begin{array}{l}
\history(w^+)=\sigma_0\sigma_1\cdots\sigma_n\in
\S^*\, \pi_{jk} \{\alpha_j\}^*\, \alpha_i\, \S_\alpha^*\,\,
\{\pi_{ik},\pi_{jk}\},\\
\history(z^-)=\sigma_0\sigma_1\cdots\sigma_m\in
\S^*\, \pi_{j'k'} \{\alpha_{j'}\}^*\, \alpha_{i'}\, \S_\alpha^*\,\,
\{\pi_{i'k'},\pi_{j'k'}\},
\end{array}
\]
for some $\{i,j,k\}=\{1,2,3\}$ and some other values of
$\{i',j',k'\}=\{1,2,3\}$.
We want to show that $\vect{w^+} > \vect{z^-}$ in order to conclude that
$|w^+|>|z^-|$.
Let $\ell$ ($h$ resp.) be the largest integer such that $w^+_\ell$ ($z^-_h$
resp.) is strong (weak resp.). The situation is illustrated in
Figure~\ref{fig:pairebispecial}.

%
We have $m\leq\ell$. Indeed, suppose on the contrary that $m>\ell$.
We know that $\sigma_m\in\S_\pi$. Also, $\sigma_i\in\S_\alpha$ for all
$\ell+1\leq i\leq n-1$. This implies that $m\geq n$ which is a contradiction.
Hence, $h<m\leq\ell<n$.
Since $\history(w^+_\ell)\in \pi_{jk} \{\alpha_j\}^*\, \alpha_i\,
\S_\alpha^*\,\, \{\pi_{ik},\pi_{jk}\}$, from
Lemma~\ref{lem:length-lower-bound} we have that $\vect{w^+_\ell}\geq(1,1,1)$.
Then,
\[
\vect{z_m}=(0,0,0) < (1,1,1) \leq \vect{w^+_\ell} \leq \vect{w^+_m}.
\]
Using induction and Lemma~\ref{lem:partialorderchignonpreserved},
we obtain that $\vect{z^-}<\vect{w^+}$.
Then, $|w^+|>|z^-|$.
\end{proof}

\subsection{Ranges of $2$ and $3$ in  the  sequence $(p(n+1)-p(n))_n$}\label{subsec:52}
The next lemma,  whose proof  requires a  deeper understanding of the abelianized  vectors of bispecial
factors,
 will   allow  us in Section \ref{subsec:proof} to   get a more precise information concerning the
alternance  of  weak and strong bispecial factors.

\begin{figure}[h!]
\begin{center}
\includegraphics[width=1.0\linewidth]{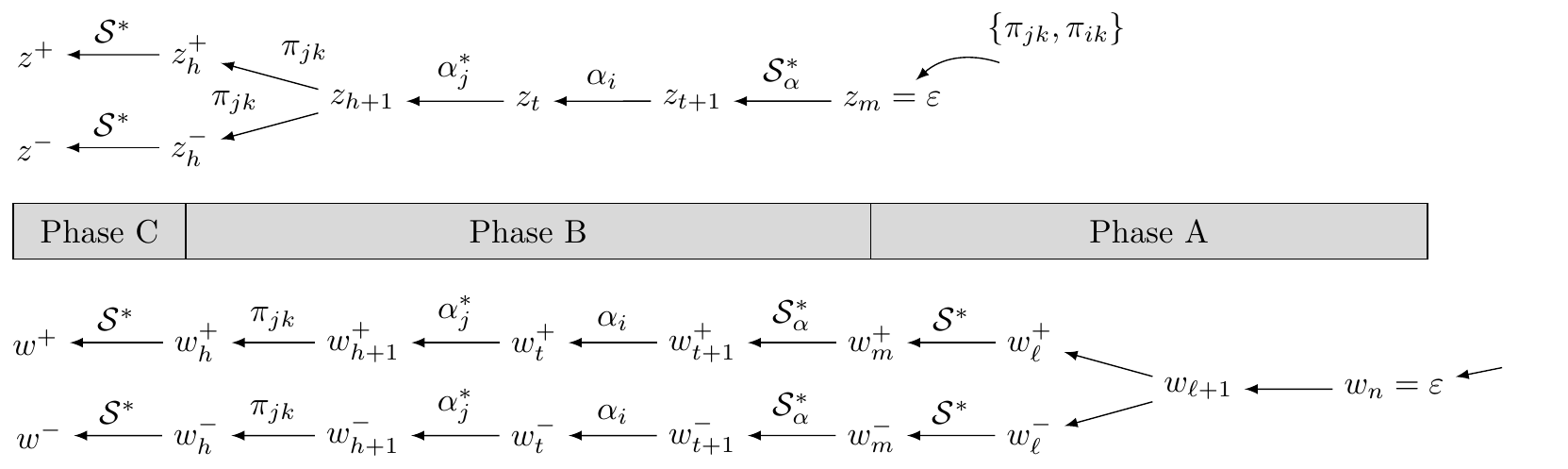}
\end{center}
\caption{Three phases of the lifes of two pairs of strong and weak bispecial factors: $z^+$,
$z^-$ and $w^+$, $w^-$.}
\label{fig:pairebispecial3phases}
\end{figure}

\begin{lemma}\label{lem:alternate_far}
Let $\bu$ be an
$\S$-adic word generated by the Arnoux-Rauzy-Poincar\'e
algorithm applied to a totally irrational vector $\bx\in\Delta$.
Let $w^+$ and $w^-$ be  two bispecial factors of the same age of   $\bu$ such
that $w^+$ is strong and $w^-$ is weak. 
If there exists a younger weak bispecial factor $z^-$ of ${\bf u}$, i.e.,
$\age(z^-)<\age(w^+)=\age(w^-)$, then $|w^+|-|z^-|>|w^-|-|w^+|$.
If there is no younger weak bispecial factor, then $|w^+|\geq|w^-|-|w^+|$.
\end{lemma}

\begin{proof}
The proof is divided into three phases according to the lifes of the
bispecial factors (see Figure~\ref{fig:pairebispecial3phases}).
\begin{enumerate}[\rm (i)]
\item At the end of Phase A, we have $\vect{w_m^+}-\vect{z_m} \geq
\vect{w_m^-}-\vect{w_m^+}$ is strict on two letters in $\{1,2,3\}$.
\item At the end of Phase B, we have $\vect{w_h^+}-\vect{z_h^-} >
\vect{w_h^-}-\vect{w_h^+}$ and the words $w_h^-$ and $w_h^+$
have the same left extensions.
\item At the end of Phase C, we have $\vect{w^+}-\vect{z^-} > \vect{w^-}-\vect{w^+}$.
\end{enumerate}

{\bf Phase A.}  Let $\ell$ be the largest index such that $w^+_{\ell} \neq
w^-_{\ell}$.
One has that $w^+_{\ell}$ is strong and $w^-_{\ell}$ is weak and their
antecent $w^+_{\ell+1}=w^-_{\ell+1}=w_{\ell+1}$ are equal and neutral.
The bispecial factor $w_\ell^+$ contains each of the letters in $\{1,2,3\}$
because of 
Lemma~\ref{lem:weak-stron-language-restricted} and
Lemma~\ref{lem:length-lower-bound}.
Also $\vect{w_\ell^-}-\vect{w_\ell^+}=\be_a$ for some $a\in\{1,2,3\}$.
Thus $\vect{w_\ell^+} \geq \vect{w_\ell^-}-\vect{w_\ell^+}$
is a strict inequality on at least two coordinates.
One checks that this property is preserved by each of the nine possible substitutions.  This implies that
$\vect{w_m^+} \geq \vect{w_m^-}-\vect{w_m^+}$
is strict on at least two letters in $\{1,2,3\}$ as well.
This also proves the last part of the lemma, concerning the case where there is no younger
weak bispecial factor.
\medskip

{\bf Phase B.}  Each of the inequality below is implied by the precedent one.
The substitution $\alpha_i$ brings the inequality (by at least two units) on
the coordinate $i$. Then, the substitution $\pi_{jk}$ spreads the strict
inequality on every coordinate:
\begin{center}
\begin{tabular}{l}
$\vect{w_m^+}-\vect{z_m} \geq \vect{w_m^-}-\vect{w_m^+}$
is strict on at least two letters in $\{1,2,3\}$,\\
$\vect{w_{t+1}^+}-\vect{z_{t+1}} \geq \vect{w_{t+1}^-}-\vect{w_{t+1}^+}$
is strict on at least two letters in $\{1,2,3\}$,\\
$\vect{w_{t}^+}-\vect{z_{t}} -2\be_i \geq \vect{w_{t}^-}-\vect{w_{t}^+}$,\\
$\vect{w_{h+1}^+}-\vect{z_{h+1}} -2\be_i \geq
\vect{w_{h+1}^-}-\vect{w_{h+1}^+}$,\\
$\vect{\pi_{jk}(w_{h+1}^+)}-\vect{\pi_{jk}(z_{h+1})} - (2,2,2) \geq
\vect{\pi_{jk}(w_{h+1}^-)}-\vect{\pi_{jk}(w_{h+1}^+)}$,\\
$\vect{w_h^+}-\vect{z_h^-} + \be_j - (2,2,2)\geq \vect{w_h^-}-\vect{w_h^+}$,\\
$\vect{w_h^+}-\vect{z_h^-} > \vect{w_h^-}-\vect{w_h^+}$.
\end{tabular}
\end{center}
The left extensions of $w_m^+$ and $w_m^-$ are $\{j,k\}$ or $\{i,k\}$.
Arnoux-Rauzy substitutions preserve the extensions so that
$E^-(w_{h+1}^+)=E^-(w_{m}^+)$ and $E^-(w_{h+1}^-)=E^-(w_{m}^-)$.
Finally, according to Table~\ref{table:poincaredegree2}, $\pi_{jk}$ projects
those left extensions onto the same set  $E^-(w_{h+1}^+)=E^-(w_{h+1}^-)=\{j,k\}$.
\medskip

{\bf Phase C.} We have $\vect{w_h^+}-\vect{z_h^-} > \vect{w_h^-}-\vect{w_h^+}$.
Since the words $w_h^-$ and $w_h^+$ have the same left
extensions, then so do $w_{h-1}^-$ and $w_{h-1}^+$ for all
$\sigma_{h-1}\in\S_\pi$. But the left extensions of $z_{h}^-$ can be different
from the one of $w_h^-$ and $w_h^+$.
This can lead to $z_{h-1}^-=jk\pi_{jk}(z_{h}^-)$ while
$w_{h-1}^-=k\pi_{jk}(w_{h}^-)$ and $w_{h-1}^+=k\pi_{jk}(w_{h}^+)$.
Thus, the proof of Phase C relies on the following recurrences on the age of
bispecial factors (all other
cases for left extensions are easier and follow from the same recurrences):
\begin{enumerate}[\rm (i)]
    \item (Recurrence \textbf{AR}) If $\vect{w^+}-\vect{z^-}>\vect{w^-}-\vect{w^+}$, then
$\vect{\alpha_{k}(w^+)}-\vect{\alpha_{k}(z^-)} >\vect{\alpha_{k}(w^-)}-\vect{\alpha_{k}(w^+)}$.
\item (Recurrence \textbf{P}) If $\vect{w^+}-\vect{z^-}>\vect{w^-}-\vect{w^+}$, then
$\vect{\pi_{jk}(w^+)}-\vect{\pi_{jk}(z^-)}-\be_j >\vect{\pi_{jk}(w^-)}-\vect{\pi_{jk}(w^+)}$.
\end{enumerate}
Let $\vect{z^-}=(x,y,z)$, $\vect{w^+}=(a,b,c)$, $\vect{w^-}=(d,e,f)$ where the
convention $\be_i=(1,0,0)$, $\be_j=(0,1,0)$, $\be_k=(0,0,1)$ is used.
For the Arnoux-Rauzy recurrence, we have
\[
\vect{\alpha_{k}(z^-)}=(x,y,x+y+z),\quad
\vect{\alpha_{k}(w^+)}=(a,b,a+b+c),\quad
\vect{\alpha_{k}(w^-)}=(d,e,d+e+f).
\]
Then
\begin{eqnarray*}
\vect{\alpha_{k}(w^+)}-\vect{\alpha_{k}(z^-)}
&=& \vect{w^+} - \vect{z^-} + (0,0,a-x) + (0,0,b-y) \\
&>& \vect{w^-} - \vect{w^+} + (0,0,a-x) + (0,0,b-y) \\
&\geq& \vect{w^-} - \vect{w^+} + (0,0,d-a+1) + (0,0,e-b+1) \\
&=& \vect{w^-} - \vect{w^+} + (0,0,d-a) + (0,0,e-b) + (0,0,2) \\
&=& \vect{w^-} - \vect{w^+} + (0,0,d-a) + (0,0,e-b) + (0,0,2) \\
&=& (d,e,d+e+f) - (a,b,a+b+c) + (0,0,2) \\
&=& \vect{\alpha_{k}(w^-)} - \vect{\alpha_{k}(w^+)} + 2\be_k
\end{eqnarray*}
For the Poincar\'e recurrence, we have
\[
\vect{\pi_{jk}(z^-)}=(x,x+y,x+y+z),\quad
\vect{\pi_{jk}(w^+)}=(a,a+b,a+b+c),\quad
\vect{\pi_{jk}(w^-)}=(d,d+e,d+e+f).
\]
Then
\begin{eqnarray*}
\vect{\pi_{jk}(w^+)}-\vect{\pi_{jk}(z^-)}-\be_j
&=& \vect{w^+} - \vect{z^-} + (0,a-x,a-x) + (0,0,b-y) - (0,1,0) \\
&>& \vect{w^-} - \vect{w^+} + (0,a-x,a-x) + (0,0,b-y) - (0,1,0) \\
&\geq& \vect{w^-} - \vect{w^+} + (0,d-a+1,d-a+1) + (0,0,e-b+1) - (0,1,0) \\
&=& \vect{w^-} - \vect{w^+} + (0,d-a,d-a) + (0,0,e-b) + (0,0,2) \\
&=& \vect{w^-} - \vect{w^+} + (0,d-a,d-a) + (0,0,e-b) + (0,0,2) \\
&=& (d,d+e,d+e+f) - (a,a+b,a+b+c) + (0,0,2) \\
&=& \vect{\pi_{jk}(w^-)} - \vect{\pi_{jk}(w^+)} + 2\be_k. \qedhere
\end{eqnarray*}
\end{proof}


\subsection{Linear  growth for the factor complexity}\label{subsec:proof}

We now  have gathered all the elements for proving Theorem~\ref{thm:leq3n}. 

\begin{proof}[Proof of Theorem~\ref{thm:leq3n}]
Since $\bx$ is totally irrational, Proposition~\ref{prop:primitive} certifies
that lemmas of  the previous two sections can be applied since  the $\S$-adic words
$\bu^{(m)}$ are  proper and uniformly recurrent  for all $m$.
The set of bispecial factors  of length $n$ contains at most one
weak or strong bispecial factor. Indeed, suppose on the contrary that it
contains two of them: $w$ and $z$.  They cannot have the same age according
to Lemma~\ref{lem:weak_larger_than_strong}  since  this would  otherwise imply
$|w|\neq|z|$.  Also, if one is older, e.g. $\age(w)>\age(z)$, then $|w|>|z|$
from Lemma~\ref{lem:alternate}. Then $b(n)\in\{-1,0,+1\}$ according to
Equation~\eqref{eq:bisp_mult} of
Theorem~\ref{thm:cassaigne454}.
Finally, it remains  to  prove that the assumptions  of
Lemma~\ref{lem:pn23iffsumbn01} are satisfied. The first non-zero value
of $b(n)$ is $+1$ because strong and weak bispecial factors come in pairs and
the strong one is smaller than the weak one from
Lemma~\ref{lem:weak_larger_than_strong}. Moreover, non-zero values are
alternating.  Indeed, let $z^+$ and $w^+$ be two strong bispecial factors such
that $\age(w^+)>\age(z^+)$. Let $z^-$ be the weak bispecial factor such that
$\age(z^-)=\age(z^+)$. From Lemma~\ref{lem:weak_larger_than_strong} and
Lemma~\ref{lem:alternate}, $|z^+|<|z^-|<|w^+|$. Hence, there is always a $-1$
between two $+1$ in the sequence $(b(n))_{n\geq0}$. This alternance of non-zero
values in the sequence $(b(n))_n$ shows that
$p(n+1)-p(n)\in\{2,3\}$ (Lemma~\ref{lem:pn23iffsumbn01}), so that $2n+1\leq p(n)\leq
3n+1$ for $n\geq0$.

Now we show that $p(n)\leq\frac{5}{2}n+1$.
We prove by recurrence that $p(q+1)\leq\frac{5}{2}(q+1)+1$ for each
$q$ such that $b(q)=-1$.
By assuming  that $b(-1)=-1$, we remark
that the statement is valid for $q=-1$ because $p(0)=1\leq1$.
Suppose $q$ and $t$ are 
two consecutive occurrences of $-1$ in the sequence $(b(\ell))_{\ell}$, that is,
$b(q)=b(t)=-1$ and $b(\ell)\neq-1$ for all $\ell$ such that $q<\ell<t$.
We show that if
    $p(q+1) \leq \frac{5}{2}(q+1) + 1$
then
    $p(n+1) \leq \frac{5}{2}(n+1) + 1$
    for each $n$ such that
$q<n\leq t$.
From the alternance of non-zero values $+1$ and $-1$ in the sequence $(b(\ell))_{\ell}$,
there exists an integer $r$ with $q<r<t$ such that $b(r)=+1$ and such that
for all integers $r'\neq r$ with $q<r'<t$ then $b(r')=0$.
Since the first non-zero value of $(b(\ell))_{\ell \geq 0}$  is $+1$, then
$\sum_{\ell=0}^{q}b(\ell)=0$.
The consequence of Lemma~\ref{lem:alternate_far} is that 
$r-q>t-r$ which is true if and only if $\frac{t-q}{2}>t-r$.
Note that if $n$ is such that $q<n\leq r$, then
$s(n)=s(0)+\sum_{\ell=0}^{n-1}b(\ell)=2+\sum_{\ell=0}^{q}b(\ell)=2$.
Also, if $n$ is such that $r<n\leq t$, then
$s(n)=s(0)+\sum_{\ell=0}^{n-1}b(\ell)=2+\sum_{\ell=0}^{q}b(\ell)+b(r)=2+1=3$.
Therefore, for each $n$ such that $r<n\leq t$ we have 
\begin{eqnarray*}
    p(n+1) - p(q+1)
    &=& \sum_{\ell=q+1}^{r}s(\ell) + \sum_{\ell=r+1}^{n}s(\ell)
    = 2(r-q) + 3 (n-r)\\
    &=& 2(n-q) + (n-r)
    < 2(n-q) + \frac{n-q}{2}
    = \frac{5}{2}(n-q).
\end{eqnarray*}	
But since
    $p(q+1) \leq \frac{5}{2}(q+1) + 1$
    we conclude that
    $p(n+1) \leq \frac{5}{2}(n+1) + 1$.
We get the same conclusion for each $n$ such that $q<n\leq r$.
From this we conclude
that $p(n)\leq\frac{5}{2}n+1$ and
$\limsup_{n\to\infty}\frac{p(n)}{n} \leq \frac{5}{2}$.
\end{proof}

We in fact prove  the more general result.
\begin{theorem} \label{thm:leq3nbis}
Let $\bu$ be a word of the Arnoux-Rauzy-Poincar\'e system.
\begin{itemize}
\item If $\bu$ is of Type $1$, then it has a bounded factor complexity.
\item If $\bu$ is of Type $2$, then its factor complexity satisfies  ultimately $p(n)=n+k$ for some  constant~$k$.
\item If  $\bu$ is of Type $3$, then $p(n+1)-p(n) \in \{2,3\}$ and $2n+1\leq p(n) \leq
    \frac{5}{2}n+1$ for all $n\geq0$.
\end{itemize}
\end{theorem}

\begin{proof}
Words $\bu$ of Type $1$ are periodic and thus  have a bounded factor
complexity.   A word $\bu$ of Type $2$ is an  image by a  substitution of  a Sturmian
sequences.
Then, according to \cite{DBLP:conf/dlt/Cassaigne97},  its  factor complexity satisfies ultimately $p(n)=n+k$ for some  constant $k$.
A word $\bu$ of Type $3$ is  weakly primitive, $\bu^{(m)}$  is recurrent and proper for all $m$,  and its   factor complexity was proven to satisfy  the  desired bounds
in Theorem~\ref{thm:leq3n}.
\end{proof}






\section{Convergence  and unique ergodicity} \label{sec:convergence}

We start with some terminology.
Let $\bu$ be an infinite word in~${\mathcal A}^{\mathbb N}$. Let $X_{\bu}$ be the orbit closure of the infinite word $\bu$ under the
action of the shift~$S$, that is,  $X_{\bu}$ is  the closure in ${\mathcal A} ^{\mathbb N}$ of the set
$\{S^n(\bu) \mid n \in \NN\} = \{(u_k)_{k \geq n} \mid k \in \NN\}$, where the shift $S$ satisfies $S((u_n)_n)=(u_{n+1})$.
The set~$X_{\bu}$  coincides with the  set of infinite words whose
language is contained in ${\mathcal L}({\bu})$, and is called the  \emph{symbolic dynamical system} generated by~$\bu$.
The topological  dynamical system $(X_{\bf u}, S)$  can be  endowed  with a  structure   of  a     measure-theoretic dynamical
      system  $(X_u, T,
    \mu,{\cal B})$, where ${\cal B}$ is a  $\sigma$-algebra,  by taking any  probability measure   $\mu$  preserved by  $T$,     that is,
     for all $B \in \mathcal{B}$, $\mu
    (S^{-1}(B))=\mu(B)$.  The system  $X_{\bu}$ is  said to be \emph{uniquely ergodic} if  there
exists a unique shift-invariant probability measure on $X$.

    One   natural  way for  getting   an $S$-invariant  measure is to consider  factor frequencies (for more details, see \cite{FM}).
    The \emph{frequency} of a letter $i$ in $\bu$ is defined as the limit when
$n$ tends towards infinity, if it exists, of the number of occurrences of $i$ in $u_0 u_1 \cdots u_{n-1}$ divided
by $n$.  The infinite word $\bu$ has \emph{uniform letter frequencies} if, for every letter $i$ of $u$, the number of
occurrences of $i$ in $u_k\cdots u_{k+n-1}$ divided by $n$ has a limit when $n$ tends to infinity, uniformly in $k$.
Similarly, we can define  the frequency   and the uniform frequency of a factor, and we say that $u$ has \emph{uniform frequencies} if all its
factors have uniform frequency. The property of  having uniform  factor frequencies  for a shift $X$ is actually equivalent to  unique ergodicity (see e.g. \cite{FM}).  

Factor complexity   is a priori a  topological  notion.   However it may yield    (in particular  when it has a  linear growth order)
measure-theoretical information on the  the  symbolic dynamical system  $X_{\bf u}$. Indeed, 
according to  \cite{Boshernitzan1984},   if $\bu$ is assumed to be uniformly recurrent,   and if  $\limsup p(n)/n <3$, then  $(X_{\bu}, S)$ is  uniquely  ergodic.

\begin{proof}[Proof of Theorem~\ref{thm:conv}]
We now have gathered all the elements for   observing that  Theorem \ref{thm:conv}  is     a direct consequence of  Theorem \ref{thm:leq3n} together with   the above mentioned result
 of   \cite{Boshernitzan1984} and Proposition~\ref{prop:primitive}. 
 \end{proof}

\section{Conclusion}


Given a  totally irrational vector of frequencies $\bx=(x_1,x_2,x_3)\in\RR^3_+$ (with $\sum
x_i=1$), we thus have shown  how  to  construct  an infinite word $\bu$
over the alphabet $\A=\{1,2,3\}$ such that the frequency of each letter
$i\in\A$ exists and is equal to $x_i$,  with  this word $\bu$  having a  linear   factor complexity. This word is contructed by  translating symbolically
within the $S$-adic formalism a multidimensional continued fraction algorithm,  namely the Arnoux-Rauzy-Poincar\'e algorithm.

 Observe that usual proofs of convergence for multidimensional continued  fraction algorithms  rely on  linear algebra
and  on the use of the Hilbert projective metric (see e.g.  \cite{SCH}).  Let
us stress the fact that we provide here a purely combinatorial  proof of
convergence for  a two-dimensional continued fraction  algorithm  based on the unique ergodicity.

The restriction to the regular language $\L(\G)$  is clearly important;    there
exist examples  of $\S$-adic words constructed  with the alphabet of
substitutions  $\S$  for which  the upper bound $\frac{5}{2}n+1$  does not  hold.
 Moreover,   a quadratic complexity is even  also achievable (see Section \ref{sec:quadratic}).  Hence, the present study  gives some more insight on a
statement of the $S$-adic conjecture (it  rather should be qualified of problem)  which is to find conditions for which
$S$-adic sequences  have a linear complexity (see e.g.
\cite{durand_do_2013,Leroy12}).
 Note that   any uniformly recurrent word $\bu$
  whose complexity function $p(n)$ satisfies $p(n+1)-p_u(n) \leq k$, for all $n$, is $S_k$-adic, with  a set  
 $S_k$ of substitutions that depends on  $k$ (\cite{frank}).
 
The upper bound $\limsup_{n\to\infty}\frac{p(n)}{n} \leq \frac{5}{2}$ is not
sharp.  Numerical experimentations tend to  indicate that the worst case in the language $\L(\G)$ of the 
Arnoux-Rauzy-Poincaré algorithm is  obtained  with the fixed point of $\pi_{23}\alpha_1$ for which
the value is approximately $\limsup_{n\to\infty}\frac{p(n)}{n} \approx
2.26079201$.




Factor complexity of Poincar\'e and Arnoux-Rauzy substitutions can be
described exactly by considering left and right extensions of length one.  It
is not always the case, and  the study of Brun substitutions (provided by the Brun multidimensional continued fraction algorithm)  seems to be an example for
which extensions of length longer than one are necessary to describe bispecial
factors. Recently, Klouda \cite{MR2928192} described bispecial factors in
fixed points of morphisms where extensions of length longer than one were
considered. Extending this work to $S$-adic words deserves further
research.

Balance  properties of the Poincar\'e and Arnoux-Rauzy $S$-adic system have also
nice properties and their study should be done more deeply. 
An infinite word $\bu \in {\mathcal A}^\NN$ is said to be $C$-balanced  if for any pair $v,w$ of factors of the
same length of $\bu$, and for any letter $i \in \mathcal{A}$, one has $||v|_i - |w|_i| \leq C$. It is said
balanced if there exists $C>0$ such that it is $C$-balanced.   For example, it was
proven in \cite{DBLP:conf/cwords/DelecroixHS13} that words generated by Brun
algorithm gives almost everywhere balanced sequences.  Balance properties are intimately connected with
Diophantine properties of the algorithm.
Indeed, an infinite word $\bu \in {\mathcal A}^\NN$ is balanced if and only if it has uniform letter frequencies and there exists
a constant $B$ such that for any factor $w$ of $u$, we have $||w|_i - f_i |w|| \leq B$ for all letter $i$ in ${\mathcal A}$,
where $f_i$ is the frequency of $i$.

\section{Appendix}

\begin{proof}[Proof of Proposition~\ref{prop:1}]
We define
\[
\tilP =
\{ A_jH_{jk}: \{i,j,k\} = \{1,2,3\} \} \cup
\{ P_{jk}H_{jk} : \{i,j,k\} = \{1,2,3\} \}
\]
which describes another partition of $\Delta$ into 12 triangles shown in
Figure~\ref{fig:partition_markov}. 

First, we show that the transformation
$T$ is a Markov transformation for the partition~$\tilP$.
\begin{figure}[h]
\begin{center}
\includegraphics[width=0.50\linewidth]{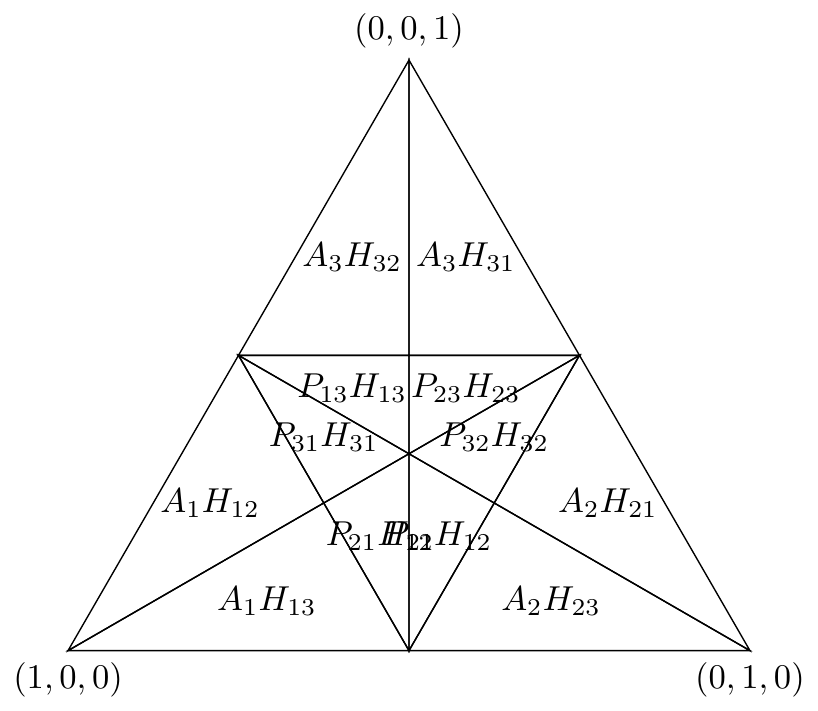}
\end{center}
\caption{The Markov partition $\tilP$ of Arnoux-Rauzy-Poincar\'e algorithm.}
\label{fig:partition_markov}
\end{figure}
Let $\{i,j,k\}=\{1,2,3\}$.
The image of $A_jH_{jk}$ and of $P_{jk}H_{jk}$ under $T$ are the
same and are equal to the half triangle $H_{jk}$:
\[
T (A_jH_{jk})
= T (P_{jk}H_{jk})
= H_{jk}.
\]
But the half triangle $H_{jk}$  is a union of elements of $\tilP$:
\[
H_{jk} =
A_iH_{ik}\cup A_iH_{ij}\cup A_jH_{jk} \cup P_{ij}H_{ij}\cup P_{ji}H_{ji}\cup P_{ki}H_{ki}.
\]
Thus, the transformation $T$ is a Markov transformation for the partition
$\tilP$.
This defines an automaton 
$\widetilde{G}=(\tilP, \widetilde{\Sigma}, \widetilde{\delta},\widetilde{I},
\widetilde{F})$
where the alphabet is
\[
    \widetilde{\Sigma} = \{
A_1^{-1}, A_2^{-1}, A_3^{-1}, P_{31}^{-1}, P_{13}^{-1}, P_{23}^{-1},
P_{32}^{-1}, P_{12}^{-1}, P_{21}^{-1} \},
\]
the transitions are 
\[
\widetilde{\delta} = \{(p,M,q)\in
\tilP\times\widetilde{\Sigma}\times\tilP : q\subseteq M\cdot p=T(p)\},
\]
or, more precisely,
\[
    \widetilde{\delta} = \left\{
\begin{array}{ll}
A_jH_{jk},A_j^{-1} \to A_iH_{ik},    & P_{jk}H_{jk},P_{jk}^{-1} \to A_iH_{ik},    \\
A_jH_{jk},A_j^{-1} \to A_iH_{ij},    & P_{jk}H_{jk},P_{jk}^{-1} \to A_iH_{ij},    \\
A_jH_{jk},A_j^{-1} \to A_jH_{jk},    & P_{jk}H_{jk},P_{jk}^{-1} \to A_jH_{jk},    \\
A_jH_{jk},A_j^{-1} \to P_{ij}H_{ij}, & P_{jk}H_{jk},P_{jk}^{-1} \to P_{ij}H_{ij}, \\
A_jH_{jk},A_j^{-1} \to P_{ji}H_{ji}, & P_{jk}H_{jk},P_{jk}^{-1} \to P_{ji}H_{ji}, \\
A_jH_{jk},A_j^{-1} \to P_{ki}H_{ki}, & P_{jk}H_{jk},P_{jk}^{-1} \to P_{ki}H_{ki}
\end{array}
\text{for each }\, \{i,j,k\} = \{1,2,3\}
\right\},
\]
the initial states and final states are all of the twelve states, i.e.,
$\widetilde{I}=\widetilde{F}=\tilP$. The automaton $\widetilde{G}$ recognize all  the expansions of the
Arnoux-Rauzy-Poincar\'e continued fraction algorithm.  It is clearly not
deterministic. A minimized and deterministic version of it is the automaton
$\G$ shown in Figure~\ref{figure:markovchain} where the alphabet considered is
$
\S = \{
\alpha_1, \alpha_2, \alpha_3, \pi_{31}, \pi_{13}, \pi_{23},
\pi_{32}, \pi_{12}, \pi_{21} \}
$
instead of $\widetilde{\Sigma}$.
In fact, amongst all the elements of $2^{\tilP}$ considered in the determinization
process, only the states in the set  $Q=\{\Delta, H_{12}, H_{13}, H_{21}, H_{23},
H_{31}, H_{32}\}$ survive the minimization.
\end{proof}

\begin{proof}[Proof of Lemma~\ref{lem:partialorderpreserved}]
(i)
Let $\vect{z}=(z_1,z_2,z_3)=\vect{\alpha_{k}(v)}$.
Let $\vect{z'}=(z'_1,z'_2,z'_3)=\vect{\alpha_{k}(v')}$.
We have
\[
\left\{
\begin{array}{l}
z_i = v_i\\
z_j = v_j\\
z_k = v_i + v_j + v_k
\end{array}
\right.
\quad
\text{and}
\quad
\left\{
\begin{array}{l}
z'_i = v'_i\\
z'_j = v'_j\\
z'_k = v'_i + v'_j + v'_k.
\end{array}
\right.
\]
Then
\[
z_k+2 = v_i+v_j+v_k + 2 = (v_i+1)+(v_j+1)+(v_k+1)-1 \leq v'_i+v'_j+v'_k-1 < z'_k
\]
and
$\vect{z}+2e_k<\vect{z'}$.\\
(ii)
Let $\vect{z}=(z_1,z_2,z_3)=\vect{\pi_{jk}(v)}$.
Let $\vect{z'}=(z'_1,z'_2,z'_3)=\vect{\pi_{jk}(v')}$.
We have
\[
\left\{
\begin{array}{l}
z_i = v_i\\
z_j = v_i + v_j\\
z_k = v_i + v_j + v_k
\end{array}
\right.
\quad
\text{and}
\quad
\left\{
\begin{array}{l}
z'_i = v'_i\\
z'_j = v'_i + v'_j\\
z'_k = v'_i + v'_j + v'_k
\end{array}
\right.
\]
As above we have $z_k+2 < z'_k$. Moreover,
\[
z_j+1 = v_i+v_j + 1 = (v_i+1)+(v_j+1)-1 \leq v'_i+v'_j-1 < z'_j.
\]
Then $\vect{z}+e_j+2e_k<\vect{z'}$.
\end{proof}

\begin{proof}[Proof of Lemma~\ref{lem:partialorderchignonpreserved}]
(i) Under Arnoux-Rauzy substitution, the extended image of $v$ and $v'$ are
uniquely determined: $w=k\alpha_k(v)$ and $w'=k\alpha_k(v')$.
From Lemma~\ref{lem:partialorderpreserved},
$\vect{\alpha_k(v)}<\vect{\alpha_k(v')}$.
Then
\[
\vect{w}=\vect{\alpha_k(v)}+e_k<\vect{\alpha_k(v')}+e_k=\vect{w'}.
\]
(ii)
The proof is divided into four disjoint cases depending on the values of
$w\in\{jk\pi_{jk}(v),k\pi_{jk}(v)\}$
and $w'\in\{jk\pi_{jk}(v'),k\pi_{jk}(v')\}$.
The proof relies on the fact that
$\vect{\pi_{jk}(v)}< \vect{\pi_{jk}(v')}$
but only the
fourth case makes a stronger use of
Lemma~\ref{lem:partialorderpreserved}, i.e.,
$\vect{\pi_{jk}(v)}+e_j<\vect{\pi_{jk}(v')}$.

(ii.i) If $w=jk\pi_{jk}(v)$ and $w'=jk\pi_{jk}(v')$, then
\[
\vect{w}=\vect{\pi_{jk}(v)}+e_j+e_k< \vect{\pi_{jk}(v')}+e_j+e_k=\vect{w'}.
\]
(ii.ii) If $w=k\pi_{jk}(v)$ and $w'=k\pi_{jk}(v')$, then
\[
\vect{w}=\vect{\pi_{jk}(v)}+e_k< \vect{\pi_{jk}(v')}+e_k=\vect{w'}.
\]
(ii.iii) If $w=k\pi_{jk}(v)$ and $w'=jk\pi_{jk}(v')$, then
\[
\vect{w}=\vect{\pi_{jk}(v)}+e_k< \vect{\pi_{jk}(v')}+e_j+e_k=\vect{w'}.
\]
(ii.iv) If $w=jk\pi_{jk}(v)$ and $w'=k\pi_{jk}(v')$, then
\[
\vect{w}=\vect{\pi_{jk}(v)}+e_j+e_k< \vect{\pi_{jk}(v')}+e_k=\vect{w'}.\qedhere
\]
\end{proof}

\bibliographystyle{alpha}
\bibliography{biblio}

\end{document}